\RequirePackage{ifthen}
\newcommand{\typeof}{1} %
\newcommand{\longversion}[1]{\ifthenelse{\equal{\typeof}{0}}{}{#1}}
\newcommand{\shortversion}[1]{\ifthenelse{\equal{\typeof}{0}}{#1}{}}
\newcommand{\longshortversion}[2]{\ifthenelse{\equal{\typeof}{0}}{#2}{#1}}
\documentclass[a4paper,UKenglish,english, cleveref, autoref, thm-restate]{lipics-v2021}

\pdfoutput=1

\newcommand{\Set}{\mathrm{Set}}

\newcommand{\Met}{\mathrm{Met}}

\newcommand{\UMet}{\mathrm{UMet}}
\newcommand{\PUMet}{\mathrm{PUMet}}

\newcommand{\InjMet}{\mathrm{InjMet}}

\newcommand{\PReal}{\BB R_{\geq 0}^{\infty}}
\newcommand{\INT}{\mathrm{IntST}}

\newcommand{\T}{T}

\newcommand{\B}[1]{\mathbf{#1}}

\newcommand{\C}[1]{\mathcal{#1}}
\newcommand{\BB}[1]{\mathbb{#1}}

\newcommand{\Ev}{\mathbf{Ev}}
\newcommand{\LL}{\boldsymbol{\Lambda}}

\newcommand{\To}{\Rightarrow}

\newcommand{\CL}{\B{CL}}
\newcommand{\III}{\mathsf{I}}
\newcommand{\KKK}{\mathsf{K}}
\newcommand{\SSS}{\mathsf{S}}
\newcommand{\YYY}{\mathsf{Y}}
\newcommand{\OOO}{\mathsf{O}}

\newcommand{\refl}{$\mathrm{(refl)}$}
\newcommand{\symm}{$\mathrm{(symm)}$}
\newcommand{\trans}{$\mathrm{(trans)}$}
\newcommand{\transs}{$\mathrm{(trans^{*})}$}

\newcommand{\refll}{$\mathrm{(refl^{*})}$}

\newcommand{\UCL}{\C U_{\CL}}

\newcommand{\free}{\mathsf{fv}}
\newcommand{\bd}{\mathsf{bd}}
\newcommand{\avar}{\mathsf{var}}

\newcommand{\FreeAlgS}[2]{\BB F_{#1}(#2)}

\newcommand{\FreeAlgLL}[1]{\BB F_{\lambda}^{\simeq}(\Var)_{#1}}

\newcommand{\FreeAlg}[1]{\BB F_{\Sigma}(#1)}

\newcommand{\Var}{\mathrm{Var}}

\newcommand{\vLambda}{\vdash^{\boldsymbol\lambda}}
\newcommand{\vLambdaEta}{\vdash^{\boldsymbol\lambda\boldsymbol \eta}}

\newcommand{\dLambda}[1]{d^{\boldsymbol\lambda}}
\newcommand{\dLambdaEta}[1]{d^{\boldsymbol\lambda\boldsymbol\eta}}

\newcommand{\VLambda}{\C V^{\boldsymbol\lambda}}
\newcommand{\VLambdaEta}{\C V^{\boldsymbol\lambda\boldsymbol\eta}}

\newcommand{\Weak}{\stackrel{\mathrm w}{\To}}
\newcommand{\VWeak}{\stackrel{\mathrm{vw}}{\To}}

\newcommand{\EQU}[3]{\stackrel{#2 #3 }{\simeq_{#1}}}
\newcommand{\EQUL}[3]{\stackrel{#2, #3 }{\simeq_{#1}}}

\newcommand{\p}[1]{\mbox{$[\![#1]\!]$}}
\newcommand{\ie}{\emph{i.e.}}
\newcommand{\eg}{\emph{e.g.}}

\newcommand{\dnf}[1]{\mbox{$d^\mathit{NF}_{#1}$}}
\newcommand{\dd}[1]{\mbox{$d^\mathit{\cal D}_{#1}$}}

\usepackage{soul}
\usepackage{adjustbox}
\usepackage{enumitem}
\usepackage{tikz}
\usetikzlibrary{cd}

\nolinenumbers

\usepackage{nameref}

\makeatletter
\let\orgdescriptionlabel\descriptionlabel
\renewcommand*{\descriptionlabel}[1]{%
  \let\orglabel\label
  \let\label\@gobble
  \phantomsection
  \edef\@currentlabel{#1}%
  \let\label\orglabel
  \orgdescriptionlabel{#1}%
}
\makeatother

\hideLIPIcs  

\title{On Quantitative Algebraic Higher-Order Theories} 

\author{Ugo Dal Lago}{Dipartimento di Informatica-Scienza e Ingegneria, Universit\`a di Bologna}{ugodallago@unibo.it}{}{}
\author{Furio Honsell}{Dipartimento di Scienze Matematiche, Informatiche e Fisiche, Universit\`a di Udine}{honsell@dimi.uniud.it}{}{}
\author{Marina Lenisa}{Dipartimento di Scienze Matematiche, Informatiche e Fisiche, Universit\`a di Udine}{marina.lenisa@uniud.it}{}{}
\author{Paolo Pistone}{Dipartimento di Informatica-Scienza e Ingegneria, Universit\`a di Bologna}{paolo.piston2@unibo.it}{}{}

\authorrunning{U.~Dal Lago, F.~Honsell, M.~Lenisa, P.~Pistone}

\Copyright{U.~Dal Lago, F.~Honsell, M.~Lenisa, P.~Pistone}


\category{} 

\relatedversion{} 

\EventEditors{John Q. Open and Joan R. Access}
\EventNoEds{2}
\EventLongTitle{42nd Conference on Very Important Topics (CVIT 2016)}
\EventShortTitle{CVIT 2016}
\EventAcronym{CVIT}
\EventYear{2016}
\EventDate{December 24--27, 2016}
\EventLocation{Little Whinging, United Kingdom}
\EventLogo{}
\SeriesVolume{42}
\ArticleNo{23}

\begin{document}

\maketitle

\begin{abstract}
We explore the possibility of extending Mardare {\em et al} quantitative algebras to the structures which naturally emerge from Combinatory Logic and the $\lambda$-calculus.  First of all, we show that the framework is indeed applicable to those structures, and give soundness and completeness results. Then, we prove some negative results which clearly delineate to which extent categories of metric spaces can be models of such theories. We conclude by giving several examples of non-trivial higher-order quantitative algebras.\end{abstract}

\section{Introduction}\label{section1}
One way of seeing program semantics is as the science of program equivalence. 
Each way of giving semantics to programs implicitly identifies which programs 
are equivalent. Similarly, a notion of program equivalence can be seen as  a way of 
attributing meaning to programs (namely, the equivalence class to which the 
program belongs). This point of view makes semantics a powerful source of ideas and techniques 
for program transformation and program verification, with 
the remarkable advantage that such techniques 
can be defined in a 
compositional and modular way.

However, there are circumstances in which equivalences between programs, being purely
{dychotomous}, 
are just not informative enough: two programs are either equivalent or 
not, period.
No further quantitative or causal 
information can be extracted from two programs which are \emph{slightly} different, although not equivalent. Furthermore, as  program equivalences are usually congruences,
 and therefore preserved by any context, programs that only differ in 
\emph{peculiar} circumstances are also just non-equivalent ones. For these reasons, methods 
alternative to program equivalence have to be looked for in all (very common) 
situations involving  transformations that replace a program by one which is 
only \emph{approximately} equivalent~\cite{Mittal2016}, or when the 
specifications are either not precise or not to be met precisely (e.g. in 
modern cryptography~\cite{KatzLindell2014}, in which most security properties 
hold in an approximate sense, namely modulo a negligible probability).

The considerations above led the scientific community to question the 
possibility of broadening the scope of program semantics from a science of equivalences to a 
science of \emph{distances} between programs. By the way, the possibility of interpreting programs 
in domains having a metric structure has been known since the 
1990s~\cite{deBakker92,deBakker96}. Recently, Mardare, Panangaden, and Plotkin have 
introduced a notion of quantitative algebra~\cite{Plotk} that generalizes usual equational reasoning to a setting in which the compared entities can be at a 
certain distance. In this way, various notions of quantitative algebra have been shown to be 
captured through a formal system, \emph{\`a la} Birkhoff~\cite{birkhoff_1935}.

Still, when the programs at hand are higher-order functional programs, the construction of a metric 
semantics faces several obstacles. First, it is well-known that 
the category $\Met$ of metric spaces and non-expansive maps, providing the standard setting of the approaches just recalled,  is \emph{not} a model of the simply typed $\lambda$-calculus (more precisely, it is not cartesian closed).
Furthermore, finding relevant \emph{sub-}categories of $\Met$ enjoying enough structure to model higher-order programs can lead to trivial (i.e.~discrete) models, and several (mostly negative) results have remained so far in the \emph{folklore} (with a few notable exceptions, e.g.~\cite{Escardo1999}). 



In this paper we bite the tail of the dragon: 
we apply 
the framework of quantitative equational theories and 
algebras from Mardare {\em et al} to the cases of  
combinatory logic and the $\lambda$-calculus, and we try to highlight features and obstacles in the construction of higher-order quantitative algebras, at the same time showing the existence of   
several interesting models. 

There are various reasons for exploring combinatory algebras, \ie\ applicative structures where the $\xi$-rule fails. The first is that these structures naturally arise in various contexts, most notably in Game Semantics and in particular in the ``Geometry of Interaction''-construction as introduced by Girard and Abramsky, see \cite{DBLP:journals/mscs/AbramskyHS02}. The $\xi$-rule can then be enforced only by introducing a rather complex notion of equivalence relation, whose fine structure is usually rather awkward to grasp. The second reason is that combinatory algebras, being indeed algebras,  might appear at first sight to be amenable straightforwardly in the first order framework of quantitative algebras of Mardare {\em et al}. We show that this is illusory, because the impact of the basic assumption that constructors are non-expansive, \ie\ Axiom NExp (see Section \ref{section2tris2}) is very strong, even in a context which could appear to be algebraically well-behaved. Finally, even if it is convenient to assume the $\xi$-rule, in reasoning on higher-order programming languages, showing that it holds in implementations is not at all immediate and, when side-effects are present, it needs to be carefully phrased.

 The contributions of this paper are threefold: \\
(i)   We introduce, following the framework defined by Mardare {\em et al}, quantitative 
   generalizations of the standard notions of \emph{weak $\lambda$-theories} and 
   \emph{$\lambda$-theories} \cite{Barendregt92}, as well as their corresponding notions of algebra.
      This is in 
   Section~\ref{sectionM-S}, Section~\ref{section2tris2}, and Section~\ref{section3tris}, 
   respectively.\\
   (ii) We study properties and examples of algebras for such theories, as suitable sub-categories of $\Met$.
   In particular, we highlight the relevance of \emph{ultra-}metric and \emph{injective} metric spaces in the construction of non-trivial (i.e.~non discrete) algebras.
 Some examples are discussed through Section~\ref{section1plus} and Section~\ref{section3tris}, further properties and examples are in Section~\ref{section4}.\\
%
%
%
(iii) Finally, we discuss algebras obtained by relaxing the conditions from Mardare {\em et al}: either by 
	replacing metrics by \emph{partial} metrics \cite{matthews, Stubbe2014}, \ie\ generalized metrics in which self-distances $d(x,x)$ need not be zero, or by relaxing the non-expansiveness condition and introducing a class of 	\emph{approximate} quantitative algebras.
%
This is in Section~\ref{section5} and Section~\ref{section7}.

\shortversion{
A longer version of this paper containing all proofs is available at the address \url{addArXivURL}.
}

\section{Preliminaries on Metric Spaces}\label{section1plus}

In this section we discuss a few properties of metric spaces and their associated categories, which provide the general setting for quantitative algebras in the sense of Mardare {\em et al}
In particular, we recall the definition of ultra-metric spaces, as well as \emph{partial ultra-metric spaces} \cite{matthews, Stubbe2014}. The latter is a class of generalized metric spaces in which self-distances $a(x,x)$ are not required to be $0$ but only smaller than any distance of the form $a(x,y)$.

\begin{definition}\label{def:metricspaces}

A pair $(X,a)$ formed by a set $X$ and a function
$a: X\times X\to \PReal$ is called:\\
(i) a \emph{pre-metric space} if it satisfies, for all $x,y\in X$, $a(x,x)=0$ \refl~and $a(x,y)=a(y,x)$ \symm; 
(ii) a \emph{(pseudo-)metric space} if it satisfies \refl, \symm, and, for all $x,y,z\in X$, $a(x,y)\leq a(x,z)+a(z,y)$ \trans; 
%
(iii) an \emph{ultra-metric space} if it satisfies \refl, \symm~and, for all $x,y,z\in X$, $a(x,y)\leq \max\{a(x,z),a(z,y)\}$ \transs; 
(iv)
a \emph{partial ultra-metric space} if it satisfies \symm, \transs~and, for all $x,y\in X$,
$a(x,y)\geq a(x,x), a(y,y)$ \refll.
%
\end{definition}
Since all metrics we consider are ``pseudo'', from now on we will omit this prefix.
Observe that an ultra-metric space is also a metric space.
Moreover, a partial ultra-metric space $(X,a)$ also yields an ultra-metric space $(X,a^{*})$, with $a^{*}(x,y)= 0$ if $x=y$ and $a^{*}(x,y)=a(x,y)$ otherwise.
Usually, partial metric spaces are defined using a stronger version of the 
triangular law, given by $a(x,y)\leq a(x,z)+a(z,y) - a(z,z)$. However, for partial \emph{ultra-}metrics this condition is equivalent to \transs~(see e.g.~\cite{Stubbe2014}).

The standard morphisms between metric (ultra-metric, partial ultra-metric) spaces $(X,a)$ and $(Y,b)$ are the \emph{non-expansive} functions, i.e.~those functions
$f:X\to Y$ such that for all $x,y\in X$, $b(f(x),f(y))\leq a(x,y)$.
%
%
We let $\Met$ (resp.~$\UMet$, $\PUMet$) indicate the category of metric spaces
(resp. ultra-metric spaces, partial ultra-metric spaces) and non-expansive maps.
All categories $\Met,\UMet$ and $\PUMet$ are cartesian, the product of $(X,a)$ and $(Y,b)$ being given by $(X\times Y, \max\{a,b\})$.
In $\UMet$ and $\PUMet$ the cartesian functors $\{-\}\times X$ have right-adjoints given, respectively, by 
$(\UMet(X,\{-\}), \Phi_{a,\{-\}})$ and $(\PUMet(X,\{-\}), \Phi_{a,\{-\}})$, where for all metric space $(Y,b)$,
$\Phi_{a,b}(f,g)  = \sup \{b(f(x),g(x))\mid  x\in X\} $.
For this reason, both categories are cartesian closed.

By contrast, $\Met$ is \emph{not} cartesian closed.
Indeed, the functor $(\Met(X,\{-\}), \Phi_{a,\{-\}})$ is right-adjoint in $\Met$ (and thus also in $\UMet$) to the functor $(X\times \{-\}, a+\{-\})$, but for all metric spaces $(Y,b)$, $(X\times Y, a+b)$ is isomorphic to the cartesian product $(X\times Y, \max\{a,b\})$ only when $X$ and $Y$ are ultra-metrics.
On the other hand, the exponential of $(X,a)$ and $(Y,b)$ in $\Met$, if it exists, is necessarily of the form
$(\Met(X,Y), \Xi_{a,b})$, where 
\begin{align*}
\Xi_{a,b}(f,g) & = \inf\{\delta \mid \forall x,y\in X \ \max\{\delta, a(x,y)\}\geq b(f(x),g(y))\}
\end{align*}
We use the Greek letter $\Xi$, since, as we'll see, this metric is tightly related to the interpretation of the ``$\xi$-rule'' of the $\lambda$-calculus. Notice that in general $\Xi_{a,b}$ is only a {pre-metric}. Indeed, the category of pre-metric spaces and non-expansive functions is cartesian closed, while the exponential of $(X,a)$ and $(Y,b)$ exists in $\Met$ precisely when $\Xi_{a,b}$ further satisfies \trans. 
\longversion{
\begin{lemma}
Let $(X,a)$ and $(Y,b)$ be metric spaces. 
The exponential of $(X,a)$ and $(Y,b)$ in $\Met$, if it exists, is isomorphic to
$(\Met(X,Y), \Xi_{a,b})$.
\end{lemma}
\begin{proof}
By definition $(Z,c)$ is the exponential of $(X,a)$ of $(Y,b)$ iff there exists a family of isomorphisms
$(\LL_{C},\Ev_{C}): \Met(C\times X, Y)\to \Met(C, Z)$, natural in $C$.
From this it follows that we can identify, up to bijection, $Z$ with $\Met(X,Y)$, and we can let $\LL_{C}(h)(z)(x)=h(z,x)$ and $\Ev_{C}(k)(z,x)=k(z)(x)$. 

It remains to show that $c(f,g)=\Xi_{a,b}(f,g)$, for all $f,g\in \Met(X,Y)$.
Let $\Ev_{X,Y}:= \Ev_{\Met(X,Y)}(\mathrm{id}_{\Met(X,Y)})\in \Met(\Met(X,Y)\times X, Y)$. 
Since $\Ev_{X,Y}$ is non-expansive we have, for all $f,g\in \Met(X,Y)$ and $x,y\in X$, 
$b(\Ev_{X,Y}(f )(x),\Ev_{X,Y}(g)(y)) \leq  \sup\{c(f,g),a(x,y)\}
$, which shows that $c(f,g)\geq \Xi_{a,b}(f,g)$.

For the converse direction, let 
$\delta:= \Xi_{a,b}(f,g)$ and 
$Q=(\{0,1\},d^{\delta})$ be the two points metric space given by $d^{\delta}(p,q)=\delta$. 
Let $h:Q \times X\to Y$ be the function given by $h(0,x)=f(x)=\Ev_{1}(f)(x)$ 
and $h(1,x)=g(x)=\Ev_{1}(g)(x)$.
Then for all $z,z'\in \{0,1\}$ and $x,y\in X$, $|h(z,x)-h(z',y)|\leq 
\sup\{ \delta,a(x,y)\}$, which shows that 
$h\in \Met(Q\times X, Y)$. 
Using the non-expansivity of $\LL_{Q}(h)$ 
we deduce then 
$c(f,g)  =
c(\LL_{Q}(h)(0),
\LL_{Q}(h)(1)))
\leq d^{\delta}(0,1)=\delta$, 
that is, $c(f,g)\leq \Xi_{a,b}(f,g)$.
\end{proof}

}

\begin{remark}
The distances computed with $\Xi_{a,b}$ are larger than those computed with $\Phi_{a,b}$. For example, given the functions $f,g: \BB R \to \BB R$ with $f(x)= x$ and $g(x)=x+K$, one has $\Phi_{a,b}(f,g)= K$, while
$\Xi_{a,b}(f,g)=\infty$, since  for all $\delta$, there exist $x,y$ such that $|x+K -y|> max\{|x+y|,\delta\}$.
Moreover, the distance function $\Xi_{a,b}$ should not be confused with the similar function
$
\Xi'_{a,b}(f,g)  = \inf\{\delta \mid \forall x,y\in X \   a(x,y) \leq\delta \Longrightarrow  b(f(x),g(y))\leq \delta\}$, which is also not in general a metric.
Indeed, consider the metric space consisting of the subset $\BB N$ as a subspace of $\PReal$, and the non-expansive functions $f(x)= x$ and $g(x)=x+K$. While for $K<1$ we have $\Xi'_{a,b}(f,g)=K$, we have as above that  for all $\delta$, there exist $x,y$ such that $|x+K -y|> \max\{|x+y|,\delta\}$.
\end{remark}


We will exploit the following useful characterization of exponentiable objects in $\Met$\footnote{An object $A$ in a cartesian category $\mathbb C$ is exponentiable when for all object $B$ the exponential of $B$ and $A$ exists in $\BB C$. In particular $\mathbb C$ is cartesian closed iff all its objects are exponentiable.}:
 \begin{theorem}[\cite{Clementino2006}]\label{thm:clementino}
 A metric space $(X,a)$ is exponentiable in $\Met$ iff 
 for all $x_{0},x_{2}\in X$ and $\alpha,\beta\in \PReal$ such that  
 $a(x_{0},x_{2})= \alpha+\beta$, the condition below holds:
  \begin{align}
 \forall \epsilon >0 \   \exists x_{1}\in X \text{ s.t. } a(x_{0},x_{1})< \alpha+\epsilon \text{ and } a(x_{1},x_{2})< \beta+\epsilon
 \tag{$*$}\label{star}
 \end{align}
 \end{theorem}
 Condition \eqref{star} intuitively requires $X$ to have ``enough points''. For example,
the set $\BB N$, as a subspace of $\BB R$, is not exponentiable in $\Met$ (take
$x_{0}=0, x_{1}=1$ and $\alpha=\beta=1/2$: a point between $0$ and $1$ is ``missing''). 
 Instead, condition \eqref{star} always holds when $(X,a)$ is \emph{injective} (see \cite{Espinola:2001aa, Clementino2006}): for any collection of points $\{x_{i}\}_{i\in I}$ in $X$ and positive reals $\{r_{i}\}_{i\in I}$ such that $a(x_{i},x_{j}) \leq r_{i}+r_{j}$, there is a point lying in the intersection of all balls $B(x_{i},r_{i})$.
This implies that the sub-category $\InjMet$ of $\Met$ formed by injective metric spaces is cartesian closed. 
Since the Euclidean metric is injective, there is a cartesian closed sub-category of $\Met$ formed by ``simple types'' over closed real intervals, that we'll use as working example.

\begin{example}\label{ex:reals}
Let $\INT$ be the set of \emph{simple types over the intervals}, defined by $[a,b]\in \INT$, for all intervals $[a,b]$ (with $a,b\in \PReal$ and $a\leq b$) and $i,j\in \INT \ \To  \ (i\times j) ,(i\to j)\in \INT$.
For any $i\in \INT$, the metric spaces $(\C I_{i}, d^{\C I}_{i})$ are defined by 
$\C I_{[a,b]} :=[a,b] $, 
$\C I_{i\times j}:= \C I_{i}\times \C I_{j} $, 
$\C I_{i\to j}:= \Met(\C I_{i}, \C I_{j}) $,  
$d^{\C I}_{[a,b]}(x,y):= |x-y|$, 
$d^{\C I}_{i\times j}:= \max\{d^{\C I}_{i},d^{\C I}_{j} \} $ and 
$ d^{\C I}_{i\to j}:= \Xi_{d^{\C I}_{i},d^{\C I}_{j}}$.
\end{example}
%



\section{Many-Sorted Quantitative Theories and Algebras} \label{sectionM-S}

In this section we introduce quantitative theories and algebras in the sense of \cite{Plotk}.
In order to cover both the typed and the untyped case, we consider \emph{many-sorted} theories and algebras, hence combining the quantitative (but one-sorted) approach from \cite{Plotk} with the qualitative (but many-sorted) approach from \cite{Meseguer1985}.
 \medskip

\noindent {\bf Notation.} 
For any set $I$, an \emph{$I$-sorted set} $A$ is an $I$-indexed family of sets $A=(A_{i})_{i\in I}$ (i.e.~an object of $\mathsf{Set}^{I}$), and an \emph{$I$-sorted function} $f:A\to B$ between $I$-sorted sets is an $I$-indexed family of functions 
$f=(f_{i}:A_{i}\to B_{i})_{i\in I}$ (i.e.~a morphism in $\mathsf{Set}^{I}(A,B)$).
For a set $I$, we denote by $I^*$ the set of all finite (possibly empty) lists of elements of $I$, we let $w$ range over $I^*$ and use $*$ for concatenation.
For $A$ an $I$-sorted set and $w=i_{1}\dots i_{k}\in I^{*}$, we let $A_{w}:= \prod_{j=1}^{k}A_{i_{j}}$.
We let $\Var$ denote a distinguished $I$-sorted set such that for all $i\in I$, $\Var_{i}$ is a countably infinite set of variables.
For any $I$-sorted set $A$, $I$-sorted function $f:\Var \to A$, pairwise disjoint variables $x_{1},\dots, x_{n}$, with $x_{j}\in \Var_{i_{j}}$ and $a_{1},\dots, a_{n}$ with $a_{j}\in A_{i_{j}}$, we let
$f_{\vec x, \vec a}:\Var \to A$ indicate the $I$-sorted function mapping $x_{j}$ to $a_{j}$ and being as $f$ on all other variables.
\medskip

\begin{definition}[Many-Sorted Signature]
An \emph{$I$-sorted signature} $\Sigma$ is an $I^{*}\times I$-sorted set $\{\Sigma_{w,i}\mid w\in I^{*},i\in I\}$ (i.e.~an object of $\Set^{I^{*}\times I}$).
\end{definition}

The objects $\sigma\in \Sigma_{w,i}$ will be called \emph{symbols} of the signature.

\begin{definition}[$\Sigma$-Algebra]
 A \emph{$\Sigma$-algebra} is a pair $(A, \Omega^{A})$ where $A$ is a $I$-sorted family and $\Omega^{A}$ associates each symbol $\sigma\in \Sigma_{w,i}$ with a function $\sigma_{A}: A_{w}\to A_{i}$, where 
$A_{w}=A_{i_{1}}\times\dots \times A_{i_{k}}$, for $w=i_{1}\dots i_{k}$.
 For any object $A$ of $\Set^{I}$, the \emph{free $\Sigma$-algebra over $A$}, noted $\FreeAlg{A}$, is the $I$-sorted set defined 
by the following conditions:  
(i)
 for all $x\in A_{i}$, $x\in \FreeAlg{A}_{i}$;
 (ii)
for all $\sigma \in \Sigma_{w,i}$ and $v_{1}\in \FreeAlg{A}_{w(1)},\dots, v_{k}\in \FreeAlg{A}_{w(k)}$, then {$\sigma_{\FreeAlg{A}}(v_{1},\dots, v_{k}):=\sigma(v_{1},\dots, v_{k})\in \FreeAlg{A}_{i}$.}

\end{definition}

Intuitively, $\FreeAlg{A}_{i}$ is the set of ``terms of sort $i$ with parameters in $A$''.
 Free algebras enjoy the following universal property:

\begin{proposition}\label{prop:substi} 
For any $\Sigma$-algebra $(A, \Omega^{A})$  and map $f\in \Set^{I}(B,A)$ there exists a unique $\Sigma$-homomorphism
$f^{\sharp}: \FreeAlg{B} \to A$ extending $f$, that is, such that $f= f^{\sharp}\circ \eta_{B}$, where   $\eta_{B}: B\to \FreeAlg{B}$ is the inclusion map.
\end{proposition}

Given a function $f\in \Set^{I}(B,A)$, if $t\in \FreeAlg{B}_{i}$ is some term of sort $i$ with parameters $b_{1},\dots, b_{n}$ in $B$, $f^{\sharp}t\in \FreeAlg{A}_{i}$ is the result of ``substituting'' each parameter $b_{i}$ in $t$ with $f(b_{i})$.

%
%
%
%
%

Let us now introduce the equational language of quantitative theories.   

\begin{definition} Let $\Sigma$ be an $I$-sorted signature.
\\ (i) A \emph{quantitative $\Sigma$-equation} over $ \FreeAlg{\Var}$  is an expression of the form
$
 t \EQU{\epsilon}{}{i}s
$, 
where $i\in I$, $t,s\in\FreeAlg{\Var}_{i}$
 and $\epsilon \in \BB Q_{\geq 0}$.
\\ (ii) 
For all $\epsilon \in \BB Q_{\geq 0}$, let $\C V(\Var)$ be the set of \emph{indexed $\Sigma$-equations} of the form $x \EQU{\epsilon}{}{i}y$, for some $i\in I$ and  $x,y\in \Var_{i}$, and $\C V(\FreeAlg{\Var})$ be the set of indexed $\Sigma$-equations of the form $t\EQU{\epsilon}{}{i} s$, where  $i\in I$ and 
$t,s\in (\FreeAlg{\Var})_{i}$.

\end{definition}

\begin{definition}\label{def:deducibility1}\label{def:deducibility}
A \emph{consequence relation} on the free $\Sigma$-algebra $\FreeAlg{\Var}$ is a relation $\vdash \ \subseteq\wp(\C V(\FreeAlg{\Var}))\times \C V(\FreeAlg{\Var}) $ closed under all instances of the following rules (where $\epsilon,\delta$ vary over all $\mathbb Q_{\geq 0}$): 
\begin{description}
\item[(Cut)\label{Cut}] if $\Gamma \vdash \phi$ for all $\phi \in \Gamma' $ and $\Gamma '\vdash \psi$, then
$\Gamma \vdash \psi$;
\item[(Assumpt)] if $\phi\in \Gamma$, then $\Gamma \vdash \phi$;
\item[(Refl)] $\emptyset \vdash  t \EQU{0}{}{i} t $;
\item[(Symm)]  $\{ t\EQU{\epsilon}{}{i} s\} \vdash  s \EQU{\epsilon}{}{i}t$;
\item[(Triang)] $\{  t\EQU{\epsilon}{}{i} s, s\EQU{\delta}{}{i}u \} \vdash t\EQU{\epsilon+\delta}{}{i} u$;
\item[(Max)] $\{ t \EQU{\epsilon}{}{i} s\} \vdash t \EQU{\epsilon+\delta}{}{i} s$;
\item[(Arch)] $\{  t \EQU{\delta}{}{i} s\mid \delta > \epsilon\} \vdash t \EQU{\epsilon}{}{i} s$;
\item[(NExp)\label{Nexp}] $\{  t_{1} \EQU{\epsilon}{}{i_{1}}s_{1}, \dots,  t_{k} \EQU{\epsilon}{}{i_{k}} s_{k}\} \vdash \sigma(t_{1},\dots, t_{k}) \EQU{\epsilon}{}{i} \sigma(s_{1},\dots, s_{k})$, for all $\sigma\in \Sigma_{i_1\ldots i_k,i}$;
\item[(Subst)]\label{Subst} if $f : \Var \to \FreeAlg{\Var}$, then $\Gamma \vdash t\EQU{\epsilon}{}{i} s$ implies 
$f^{\sharp}\Gamma \vdash  f^{\sharp}t \EQU{\epsilon}{}{i}  f^{\sharp}s$.



\end{description}
Notice that rule {\bf({Arch})} has infinitely many assumptions.
\end{definition}


We let $\C E(\FreeAlg{\Var})=\wp_{\mathrm{fin}}(\C V(\FreeAlg{\Var}))\times \C V(\FreeAlg{\Var})$ indicate the set of \emph{quantitative inferences on $\FreeAlg{\Var}$} and 
 $\C E(\Var)=\wp_{\mathrm{fin}}(\C V(\Var))\times \C V(\FreeAlg{\Var})$ indicate the set of \emph{basic quantitative inferences}. Axioms for theories will be basic quantitative inferences.

\begin{definition}[Many-Sorted Quantitative Theory] \label{def:theories}
Let $S\subseteq \C E(\Var)$ be a set of basic quantitative inferences.  Let $\vdash_{S}$ be the smallest consequence relation including $S$. The \emph{quantitative equational theory over $\Sigma$ generated by $S$} is the set 
$
\C U_{S} := (\vdash_{S})\cap \C E(\FreeAlg{\Var})
$. 
The elements of $S$ are the \emph{axioms} of $\C U_{S}$.
\end{definition}

To the syntactic notion of quantitative theory there corresponds a semantic notion of quantitative algebra, given by a $\Sigma$-algebra endowed with suitable metrics.

\begin{definition}[Many-Sorted Quantitative Algebra] \label{def:algebra}
Let $\Sigma$ be an $I$-sorted signature. A \emph{quantitative $\Sigma$-algebra} is a tuple $\C A=(A, \Omega^{\C A}, d^{\C A})$ where $(A,\Omega^{\C A})$ is a $\Sigma$-algebra and $d^{\C A}$ is an $I$-sorted family of metrics
$d^{\C A}_{i}: A_{i}\times A_{i} \to \BB R_{\geq 0}^{\infty}$ such that for all $\sigma\in \Sigma_{w,i}$, 
$\sigma_{A}:  \Met( A_{w}, A_{i})$.
\end{definition}

Given a quantitative $\Sigma$-algebra, we can define a multicategory $\Met^{\C A}$ whose objects are the metric spaces $(A_{i},d^{\C A}_{i})$, 
and where for all $w=i_{1}\dots i_{k}$, $\Met^{\C A}(A_{i_{1}},\dots, A_{i_{k}}; A_{i})\subseteq \Met(A_{w},A_{i})$ contains all 
 functions $f\in \Met^{\C A}(A_{w},A_{i})$ such that for some term $t_{f}\in \FreeAlg{A+\{x_{1}:w(1),\dots, x_{k}:w(k)\}}$,  
 $f(a_{1},\dots, a_{k})= f^{\sharp}_{\vec x,\vec a}(t_{f})$.
 For brevity, we will often abbreviate $\Met^{\C A}(A_{i_{1}},\dots, A_{i_{k}}; A_{i})$ as $\Met^{\C A}(A_{w};A_{i})$. 

%

 \begin{definition}\label{def:sat} { Let  $\C A = (A, \Omega^{\C A}, d^{\C A})$ be a quantitative $\Sigma$-algebra.
$\C A $ satisfies a quantitative inference $\Gamma \vdash t \EQU{\epsilon}{}{i} u $ (denoted $ \Gamma\vDash_{\C A}    t \EQU{\epsilon}{}{i} u$), if for all $f:\Var \to A$, the following holds:
\begin{align}  \text{for all } ( t' \EQU{\epsilon'}{}{i'} u' ) \in \Gamma, \ d^{\C A}_{i'} ( f^{\sharp}(t'), f^{\sharp}(u'))\leq \epsilon' \text{ implies }    d^{\C A}_i ( f^{\sharp}(t), f^{\sharp}(u))\leq \epsilon  
\tag{sat}             \label{sat}
\end{align}
$\C A $ satisfies a quantitative equational theory $\C U$ (noted $ {\C A} \vDash \C U$) if it satisfies all inferences in $\C U$.
}
\end{definition} 

Notice that the interpretation of rule (Nexp) implies that functional terms need to be interpreted as non-expansive morphisms.

\begin{remark}\label{rem:partial}
All constructions from this section can be adapted to the case of \emph{partial} ultra-metric spaces by replacing, in Def.~\ref{def:deducibility}, the rule
(Refl) with the following rule:
\begin{description} 
\item[(PRefl)] $\{ t\EQU{\epsilon}{i}{}u \}\vdash t\EQU{\epsilon}{i}{} t$;
\end{description}
 and requiring in 
Def.~\ref{def:algebra} that the $d^{\C A}_{i}$ are partial ultra-metrics and $\sigma^{\C A}\in \PUMet(A_{w},A_{i})$.
\end{remark}

\section{Quantitative Weak $\lambda$-Theories and Algebras}\label{section2tris2}

As is well-known (see e.g.~\cite{Baren85}), a purely algebraic approach to the $\lambda$-calculus is provided by \emph{combinatory logic} $\CL$. Hence,
it is  natural to start from this calculus. The equational theory of $\CL$ captures so-called \emph{weak $\lambda$-theories} \cite{Baren85}, namely $\lambda$-theories where the \emph{$\xi$-rule} (discussed in more detail in Section \ref{section3tris}) may fail. In this section we introduce quantitative weak $\lambda$-theories and we discuss their algebras, of which $\Met$ itself is a notable example.

\begin{definition}[Applicative Signature] \label{CLapp}
 Let $\T$ be a set of sorts (called \emph{types})  endowed with a binary {function} $\to: \T \times \T \to \T$.
An {\em applicative signature}  $\Sigma$ is a  $\T$-sorted signature   which includes symbols
$\cdot_{i,j}\in \Sigma_{(i\to j)*i, j}$, for all $i,j \in \T$.

\end{definition}
We will often note $\cdot_{i,j}(t,u)$ infix, \ie\ $t \cdot_{i,j}u$, or simply as $tu$, when clear from the context. For all $w=i_{1}\dots i_{n}\in \T^{*}$ and  $j\in \T$, we let $w\to j:=i_{1}\to \dots \to i_{n}\to j$.
A notable example of applicative signature is the following:
\begin{definition}[$\CL$-Signature] \label{CLsig}  
Let
$\Sigma^{\CL}$ be the applicative signature which includes symbols $\III_{i}: i\to i$, $\KKK_{ij}: i\to j\to i$,  $\SSS_{ijk}: (i\to j\to k) \to (i\to j)\to (i\to k)$, for all $i,j,k \in \T$.
%
%
%
%
The \emph{terms of combinatory logic} are the elements of the free $\Sigma^{\CL}$-algebra, $\FreeAlgS{{\CL}}{\Var}$.
%
\end{definition}

Definition~\ref{CLsig} above comprises both the typed and untyped case. In typed Combinatory Logic the set of types $\T$ includes {at least a \em{base} type $o$, \ie\ a type which is not in the image of $\to$ and $\to$ is injective}, while in the untyped case $\T$ is a singleton {set $\{\star\}$ and hence }$\star\to \star=\star$.  
 In the traditional language of ``syntax and semantics'', used for instance in \cite{Baren85}, when $f:\Var \rightarrow  A$, the function $f^\sharp$ of Proposition \ref{prop:substi}, amounts  to the notion of \emph{intepretation} of a term $t$ in  the environment $f$, namely  $f^\sharp(t)=\p{t}_f$ .

We now introduce the natural notion of theory for a $\CL$-signature:

\begin{definition}[$\CL$-Theory] The quantitative equational theory over $\FreeAlgS{{\CL}}{\Var}$, ${\C U}_{\CL}$ is generated  by the axioms 
$\emptyset  \vdash  \III_{i} t \EQU{0}{}{i} t $,  
$\emptyset  \vdash  \KKK_{ij} t u \EQU{0}{}{i} t $, and $\emptyset  \vdash  \SSS_{ijk} t u w\EQU{0}{}{k} tw(uw)$.
We call {\em (quantitative) weak $\lambda$-theory} any theory including ${\C U}_{\CL}$.
\end{definition}

\begin{example}\label{ex:reals2}
The set 
$\INT$ (cf.~Example~\ref{ex:reals}) is a particular instance of the set $\T$.
Let $\C I(\Sigma^{\CL})$ be the signature obtained by enriching $\Sigma^{\CL}$ with 0-ary symbols $\overline r \in \C I(\Sigma)_{(),[a,b]}$ for all $r\in [a,b]$, 
and $k$-ary symbols $\overline f\in \C I(\Sigma)_{[a_{1},b_{1}],\dots, [a_n,b_{n}]  , [a,b]}$ for all $f\in \Met( \prod_{i}[a_{i},b_{i}], [a,b])$.
Let $\C U_{\CL}^{\C I}$ be the theory obtained by extending $\C U_{\CL}$ with all axioms
$\emptyset \vdash \overline f\overline r_{1}\dots \overline r_{k}\EQU{0}{[a,b]}{} \overline s $ whenever 
$f(r_{1},\dots, r_{k})=s$ as well as all axioms
$\emptyset \vdash \overline r \EQU{\epsilon}{[a,b]}{} \overline s$ for all rational $\epsilon\geq |r-s|$.
\end{example}

A well-known property of Combinatory Logic is  
 \emph{functional completeness}: for any term $t$ and variable $x$, one can construct a term $\Lambda_{x}(t)$ so that $\Lambda_{x}(t)$ ``simulates'' $\lambda$-abstraction in the sense that one can prove $\Lambda_{x}(t)u \simeq t[u/x]$. 
 This leads to the following definition:
 
\begin{definition}[Quantitative Weak $\lambda$-Algebra]\label{def:weaklambda}
An applicative quantitative $\Sigma$-algebra $\C A=(A, \Omega^{\C A}, d^{\C A})$ is said a \emph{quantitative weak $\lambda$-algebra} if for all $w\in I^{*}$, $j\in I$, 
 and $f\in \Met(A_{w},A_{j})$, 
the set ${\Lambda(f)}=\{g\in A_{w\to j}\mid \forall (x_1,\ldots , x_k)\in A_{w}\ \ \ g\cdot_{A}x_1\cdot_{A} \ldots \cdot_{A}x_{k}= f(x_1,\ldots ,x_k)\}$ is non-empty.
\end{definition}


\begin{proposition}
Any  quantitative $\Sigma^{\CL}$-algebra satisfying $\C U_{\CL}$ is a quantitative weak $\lambda$-algebra. Vice versa, any quantitative weak $\lambda$-algebra satisfies $\C U_{\CL}$. 
\end{proposition}
\longversion{
\begin{proof} (sketch) For any term $t$ of combinatory logic,  one can define a term
$\Lambda^{j}_{x:i}(t)$, only depending on the variable $x$, i.e.~$\Lambda^{j}_{x:i}(t)\in \FreeAlgS{\CL}{X}_{j}$ such that 
$\emptyset \vdash \Lambda^{j}_{x:i}(t)u \simeq_{0}^{j} f^\sharp t \in \UCL$, where $F$ maps $x$ to $u$ and is the identity on the remaining variables. We simply let
$\Lambda_{x:i}^{i}(x):=\III_{i}$, $\Lambda_{x:j}^{i}(t):=\KKK_{ij} t$, where $x$ does not occur  in $t$, and
$\Lambda_{x:i}^{k}(tu):= \SSS_{ijk} \Lambda_{x:i}^{j\to k}(t)\Lambda_{x:i}^{j}(u)$. 
The reverse implication is obtained by choosing elements
$\III^{\C A}\in \Lambda(x\mapsto x)$, $\KKK^{\C A}\in \Lambda(x,y\mapsto x)$ and 
$\SSS^{\C A}\in \Lambda( f,g,x \mapsto f \cdot^{\C A} x\cdot^{\C A} (g\cdot^{\C A}x) )$.
%
%
\end{proof}
}
%

\begin{example}\label{ex:reals3}
We obtain a quantitative weak $\lambda$-algebra by letting $\C I=(\C I_{i},\Omega^{\C I},d^{\C I}_{i})$, where $\overline r^{\C I}=r$, $\overline f^{\C I}=f$, and $f \cdot^{\C I}x =f(x)$.
It is clear that 
$\C I\vDash \C U^{\C I}_{\CL}$ (cf.~Example \ref{ex:reals2}).

\end{example}

Following \cite{Martini1992}, the condition from Def.~\ref{def:weaklambda} can be specified in categorial terms: a cartesian multicategory $\BB C$ is a model of $\CL$ precisely when for all objects $A,B$ of $\BB C$ there is an object 
$A\VWeak B$ (called a \emph{very weak exponential of $A$ and $B$}) together with a surjective natural transformation
$\Phi :\BB C(\_; A\VWeak B)\to \BB C(\_ , A; B)$. When $\BB C$ is the multicategory $\Met^{\C A}$, 
the conditions of Def.~\ref{def:weaklambda} imply that $A_{i\to j}$ is a very weak exponential of $A_{i}$ and $A_{j}$ in $\Met^{\C A}$: 
a family of multiarrows $\Ev^{\C A}_{w,i,j}:\Met^{\C A}(A_{w}; A_{i\to j})\To \Met^{\C A}(A_{w*i}; A_{j})$, natural in $w$, is given by $\Ev^{\C A}_{w,i,j}(f)(z,x)=f(z)\cdot^{A}x$, and the non-emptyness of the sets $\Lambda(f)$ corresponds to the surjectivity of this transformation. 

Notice that $\Met$ itself admits very weak exponentials for all of its objects, \ie\ it is a \emph{very weak CCC} in the sense of \cite{Martini1992}, provided we endow $\Met(X,Y)$ with the metric $\Theta_{a,b}$ for metric spaces $(X,a)$ and $(Y,b)$, where for $f,g : X\to Y$
$\Theta_{a,b}(f,g)$ is $0$ if $f=g$, and otherwise is $ \sup\{b(f(x),g(y))\mid x,y \in X\} $.
%
 Intuitively, when $f\neq g$, $\Theta_{a,b}(f,g)$ measures the diameter of the interval spanned by the image of both $f$ and $g$. 
 
 \begin{remark} $\Theta_{a,b}$ and $\Xi_{a,b}$ do not coincide: 
 for example, consider $f,g: [-1,1]\to \BB R$ where $f(x)=x$ and $g(x)=x$ for $x\leq 0$ while 
 $g(x)=\frac{x}{2}$ for $x> 0$; then one can show that $\Theta_{a,b}(f,g)=2$
 while $\Xi_{a,b}(f,g)\leq \frac{\sqrt{2}}{2}$.
%
  \end{remark}
 
 The metric $\Theta_{a,b}$ is in general rather odd since the identity is an isolated point if $(X,a)$ is infinite and not trivial. Assume that the sequence $\{Id_n\}$ converges to the identity  $id$. Then we have for all $n,x,y$ the following inequality: $ a(x,y) \leq  a(Id_n(x),x) + a(Id_n(x),y)$. While $a(Id_n(x),x)$ can be made arbitrarily small, a $y$ can always be picked, provided the metric is not trivial and the space is infinite, such that $\delta\leq a(x,y)$. Thus $\Theta_{a,a}(Id_n, Id)\geq \delta$  for all $n$.
 \begin{example}
 The constructions just sketched yields a different weak $\lambda$-algebra over the reals
 $\C I_{\mathrm{weak}}= (\C I_{i}, \Omega^{\C I}, d^{\mathrm{weak}})$, where 
 $d^{\mathrm{weak}}$ is defined like $d^{\C I}$ but for
 $d^{\mathrm{weak}}_{i\to j}= \Theta_{d^{\mathrm{weak}}_{i},d^{\mathrm{weak}}_{j}}$.
 Notice that we still have $\C I_{\mathrm{weak}}\vDash \C U^{\C I}_{\CL}$, since $\C I$ and $\C I_{\mathrm{weak}}$ agree on distances of types $[a,b]$.
  \end{example}

 The result below adapts to the many-sorted case a similar result for one-sorted quantitative equational theories \cite{Plotk}. The proof is similar to that of Theorem \ref{thm:completeness2}, so we omit it.
 
\begin{theorem}[Soundness and Completeness of Quantitative Weak $\lambda$-Theories]\label{thm:completeness1}
For any quantitative weak $\lambda$-theory $\C U$ over $\Sigma^{\CL}$, 
$\Gamma \vdash \phi \in \C U$ iff 
$\Gamma \vDash_{\C A} \phi$ holds for any quantitative 
weak $\lambda$-algebra $\C A$ such that $\C A\vDash \C U$. 
\end{theorem} 


\begin{remark}\label{rem:partial2}
Following Remark \ref{rem:partial}, in the case of partial ultra-metric spaces we will talk of
\emph{partial weak $\lambda$-theories} and
\emph{partial weak $\lambda$-algebras}.
\end{remark}

\section{Quantitative $\lambda$-Theories and Algebras}\label{section3tris}

As we recalled, weak $\lambda$-theories do not fully capture the equational theory of the $\lambda$-calculus, as they fail to capture the so-called \emph{$\xi$-rule} \cite{Baren85}. 
In our quantitative setting, this rule can be expressed as the inference
$t \EQU{\epsilon}{j}{} u \vdash \lambda x.t \EQU{\epsilon}{i\to j}{} \lambda x.u $
\emph{provided} the equation on the left of $\vdash$ is \emph{locally universally quantified}: the righthand equation holds under the condition that, \emph{for all possible value of $x$}, the lefthand equation holds. This kind of quantitative inferences differ from those seen so far. The reason for this proviso is that it involves the higher-order operator $\lambda$, which ``binds'' the variable $x$.
The example below shows that quantitative weak $\lambda$-algebras  fail to capture this rule.
 
 \begin{example}\label{ex:reals4}
 The $\xi$-rule fails in the weak $\lambda$-algebra $\C I$:
let $f,g: [0,b]\to [0,b+\epsilon]$ (where $b\in \BB R_{\geq 0}$ and $\epsilon \in \BB Q_{\geq 0}$) be, respectively, the identity function $f=\mathrm{id}$ and the function $g(x)=x+\epsilon$;
for any $s\in [a,b]$, we then have $|f(s)-g(s)| \leq \epsilon$, which shows
$\C I \vDash \overline fx \EQU{\epsilon}{[a,b+\epsilon]}{} \overline gx$.
However, since %
$d^{\C I}_{[a,b]\to [a,b+\epsilon]}(f,g) = b+\epsilon$, we deduce 
$\C I \not\vDash \lambda x.\overline fx \EQU{\epsilon}{[a,b]\to [a,b+\epsilon]}{} \lambda x.\overline gx$.

 \end{example}

In order to define quantitative $\lambda$-theories we could follow Curry \cite{Baren85} and ``strengthen''  the set of axioms, in fact mere equalities, satisfied by a $\Sigma^{\CL}$-algebra and essentially do away with the $\xi$-rule and all higher order features. 
The alternative, that we develop in this section, is to take abstraction and the $\xi$-rule as first class elements of our theories and algebras. This will require a number of generalizations of the original approach of \cite{Plotk}.

%


At the level of syntax, the first step is to enrich the class of symbols with higher-order operators of the form $\lambda_{i}x$. The occurrence of the variable $x$ is part of the symbol $\lambda_i x$ itself.

\begin{definition}[$\lambda$-Signature] Given an applicative $\T$-sorted signature $\Sigma$, let $\Sigma^{\lambda}$ be the applicative $\T$-sorted signature further including the symbols
$\lambda_i x \in \Sigma^{\lambda}_{j, i\to j }$, for all $x\in \Var_i$ and $i,j\in \T$.
%
%
%
The {\em $\lambda $-terms} are the elements of the free $\Sigma^{\lambda}$-algebra, $\FreeAlgS{\lambda}{ \Var}$.

%
\end{definition}


Terms $\lambda_i x(t)$ will be denoted by $\lambda_i x.t$ or simply $\lambda x.t$.
Free and bound variables, open and closed $\lambda$-terms are defined as usual.  
For a $\lambda$-term $t$, we denote by $\free(t)$, $\bd(t)$, $\avar(t)$ the sets of free, bound, and all variables in $t$, respectively. 
 In order to simplify the notation we  deal with bound variables by implementing directly Barendregt's ``hygiene condition''.
For any function $f:\Var \to \FreeAlgS{{\lambda}}{  \Var}$ there exists a function $f^\flat: \FreeAlgS{\lambda}{ \Var}\to \FreeAlgS{\lambda}{  \Var}$ such that 
$f^{\flat}(t)$ corresponds to the
substitution of $f(x)$ for $x$ in $t$, for any variable $x$  occurring free in $t$. 
Given pairwise disjoint variables $x_{1},\dots, x_{n}$, with $x_{j}\in \Var_{j}$ and terms $t_{1},\dots, t_{n}$, with $t_{j}\in \FreeAlgS{\lambda}{\Var}_{j}$, 
we indicate the ``substitution'' $(\mathrm{id}_{\vec x, \vec t})^{\flat}(u)$ simply as
$u[t_{j}/x_{j}]$.

 
%
In order to be able to express correctly  the $\xi$-rule we generalize quantitative equations to expressions of the form
$t\EQU{\epsilon}{X}{i} u$, where $X$ indicates a finite set of variables which are intended to be  ``locally quantified'' on the left of $\vdash$. 

\begin{definition}[$\Sigma^{\lambda}$-equation]  \label{def:sigmalambda}
A \emph{quantitative ${\lambda}$-equation} is an expression of the form
$
 t \EQUL{\epsilon}{X}{i}s
$, 
where $i\in I$, $t,s\in \Lambda_{i}$,  
$X\subseteq_{\mathsf{fin}} \Var$, 
$\epsilon \in \BB Q_{\geq 0}$.  The set $X$  is  the set of \emph{locally quantified variables in the equation.}
\end{definition}

We let $\C V(\Lambda)$ indicate the set of quantitative $\lambda$-equations

\begin{definition}\label{def:deducibility1}
\hfill \\
(i) A \emph{consequence relation} on $\Lambda$ is a relation $\vdash \ \subseteq\wp(\C V(\Lambda))\times \C V(\Lambda) $ closed under the rules (Cut)-(Nexp) from Def.~\ref{def:deducibility} (with $t\EQU{\epsilon}{i}{} u$ everywhere replaced by $t\EQU{\epsilon}{X}{i} u$), together with the following rules:
%
\begin{description}
 \item[(Subst)] if $\Gamma \vdash t\EQUL{\epsilon}{X}{i} s$ and let $f $ be the identity on $X$ and, for all $x\in \Var \setminus X$, $\free(f(x))\cap \bd(t,s,\Gamma)=\emptyset$, then $\Gamma \vdash t\EQUL{\epsilon}{X}{i} s$ implies 
$f^{\flat}(\Gamma) \vdash  f^{\flat}(t) \EQUL{\epsilon}{X}{i}  f^{\flat}(s)$; 

\item[(Abstraction)] {if $X\subseteq X'$  and $\free(t,s) \cap X'=\emptyset$, then $\{ t \EQUL{\epsilon}{X}{i} s\} \vdash  \ t \EQUL{\epsilon}{X'}{i} s$}; 
\item[(Concretion)]  {if $X'\subseteq X$  then $\{ t \EQUL{\epsilon}{X}{i} s\} \vdash  \ t \EQUL{\epsilon}{X'}{i} s$;}
\end{description}

\noindent (ii) We call $\C U_{\lambda}$ the quantitative theory generated by the axioms below apart from ($\eta$) and we denote by $\vLambda$ the corresponding consequence relation, and $\C U_{\lambda\eta}$  the quantitative theory generated by all the axioms below, including ($\eta$),  with consequence relation $\vLambdaEta$:
\begin{description}

 {\item[($\alpha$)] if $x,y \in \Var_i$ and $y\not\in \avar(\lambda_i x.t)$, then $\emptyset \vdash \lambda x.t\EQUL{0}{X}i 
 \lambda_{i} y. t[y/x]$.
%
 }
\item[($\xi$)] if $x \in X$, then  $ t \EQUL{\epsilon}{X}{j} u 
 \vdash \lambda_{i}x. t \EQUL{\epsilon}{X}{i\to j} \lambda_{i}x. u$;
\item[($\beta$)] if $(\lambda_{i}x. t)u \in \Lambda_{ j}$,    $\free(u)\cap \bd(t) =\emptyset$,   
then
 {$\emptyset \vdash (\lambda_{i}x. t)u \EQUL{0}{X}{j} t[u/x]$.}
%
%
\item[($\eta $)]  if $ \lambda_{i}x.(tx)\in \Lambda_{i\to j}$,  $x\notin \free(t)$, then $ \vdash t \EQUL{0}{X}{i\to j} \lambda_{i}x.(tx) $.\smallskip

\end{description}

\noindent  (iii)  Any theory including $\C U_{\lambda}$ ($\C U_{\lambda\eta}$) is called a {\em quantitative (extensional) $\lambda$-theory}.
\end{definition}

\begin{example}\label{ex:reals5}
Consider the $\lambda$-signature $\C I(\Sigma)^{\lambda}$ 
(cf.~Example \ref{ex:reals2}).
Let $\C U^{\C I}_{\lambda\eta}$ be the extensional $\lambda$-theory obtained by enriching $\C U_{\lambda\eta}$ with all real-valued axioms as in Example \ref{ex:reals2}.

\end{example}

We now introduce a class of applicative algebras suitable to account for abstraction operators.
This is done by requiring the existence of suitable ``closing maps'' that send
a closed $\lambda$-term of the form $\lambda_{i_{1}}x_{1}.\dots.\lambda_{i_{n}}x_{n}.t$ onto
some point of $A_{i_{1}\to \dots \to i_{n}\to j}$.


 Given any $\T$-index set $A$, extend the definition of $\Sigma^{\lambda}$ to $\Sigma^{\lambda,A}$ so as to contain as 0-ary constructors all elements in $A$ and correspondingly the notion of $\FreeAlgS{\lambda, A}{\Var}$.

\begin{definition}
A {\em quantitative applicative $\lambda$-algebra} is a structure  $\C A = (A, \Omega^{\C A}, \Lambda^{\C A}, d^{\C A})$, where $(A, \Omega^{\C A},  d^{\C A})$ is a quantitative applicative algebra  and 
 $\Lambda_{w,j}^{\C A}: 
 (\FreeAlgS{{\lambda,A}}{  \Var})_{w\rightarrow j}^0 \to {A}_{i\to j}$. We call {\em applicative $\lambda$-algebra} the structure $\C A = (A, \Omega^{\C A}, \Lambda^{\C A})$ without the metric.
%
\end{definition}


The functions $\Lambda_{i,j}^{\C A}: (\FreeAlgS{\lambda, A}{  \Var})_{i\rightarrow j}^0 \to {A}_{i\to j}$ are intended to define a choice  in the set  ${\Lambda}$ of Definition \ref{def:weaklambda}.
This will be apparent in view of  Definitions \ref{interpretation},\ref{representable},\ref{quantitative-lambda} below, which will enforce  that, in suitable structures, the interpretations of the terms  $\FreeAlgS{{\lambda,A}}{  \Var})_{i\rightarrow j}^0$ become essentially the domain of ${\Lambda}$ in Definition~\ref{def:weaklambda}. We point out that a slight modification of these definitions would permit to recover precisely the categorically weaker notion of  Quantitative Weak $\lambda$-algebra of Definition~\ref{def:weaklambda}.


\begin{proposition}[Interpretation]\label{interpretation} Let  $\C A$  be a quantitative applicative $\lambda$-algebra, and $\rho: \Var \to {A}$. Then there exists a function $\rho^\natural: \FreeAlgS{{\lambda}}{\Var}\to { A} $, where $\rho^\natural (t)$ is defined by
cases as
$\rho^\natural (x)=\rho(x)$, 
$\rho^{\natural}(t_{1}\cdot t_{2})=  \rho^\natural(t_1) \cdot_A \rho^\natural(t_2) $,
and $\rho^{\natural}(\lambda_{i}x.t)  =\Lambda^{\C  A}_{w*i,j}(\lambda \vec{y}\lambda_i x.t)\cdot_A \overrightarrow{\rho^\natural(y)} $, 
where  $\lambda \vec{y}\lambda_i x.t$ is the closure of the term  $\lambda_i x.t$ w.r.t. its free variables $\vec{y}$ of types ${w}$.
\end{proposition}

To define higher-order structure for a quantitative applicative $\lambda$-algebras $\C A$ it is useful to define the multicategory generated by $\C A$:
\begin{definition}[Representable functions]\label{representable}
For any quantitative applicative $\lambda$-algebra $\C A$, $\Met^{\C A}$ is the multicategory with objects the metric spaces $(A_{i},d^{\C A}_{i})$, and where, for $w=i_{1}\dots i_{k}$, $\Met^{\C A}(A_{i_{1}},\dots, A_{i_{k}}; A_{i})$ (abbreviated as $\Met^{\C A}(A_{w};A_{i})$) is the set of $f\in \Met(A_{w},A_{i})$ such that for some $t_{f}\in \FreeAlgS{\lambda,A}{\Var}_{w\to i}^{0}$, 
$f(a_{1},\dots, a_{n})= \Lambda_{w,i}^{\C A}(t_{f}) \cdot^{\C A} 
\overrightarrow{a}
$.
\end{definition}
Notice that the function $\Lambda^{\C A}$ yields a family of maps 
$\LL^{\C A}_{w*i,j}:\Met^{\C A}(A_{w*i};A_{j}) \to \Met^{\C A}(A_{w}; A_{i\to j})$, given by $\LL_{w*i,k}(h)(a)(b)= \Lambda_{w*i,j}^{\C A}( t_{h})\cdot^{\C A}\langle a, b\rangle$. 

While cartesian closed (multi)categories are the algebras for extensional $\lambda$-theories, 
 an algebra for a 
 $\lambda$-theory 
%
  is a cartesian multicategory in which for all objects $A,B$ there is an object $A\Weak B$ (called a \emph{weak exponential}, \cite{Martini1992}) together with a natural \emph{retraction} $\BB C(\_ , A; B)\To \BB C(\_; A\Weak B)$. To account for the quantitative $\xi$-rule, this picture must be slightly adapted, by requiring the maps forming the retraction to be also \emph{non-expansive}. 
This leads to the following definition:

\begin{definition}[Quantitative $\lambda$-Algebra]\label{quantitative-lambda}
A {quantitative applicative $\lambda$-algebra } $\C A = (A, \Omega^{\C A}, \Lambda^{\C A}, d^{\C A})$ is a \emph{quantitative (extensional) $\lambda$-algebra} if the 
maps 
$\LL^{\C A}_{w*i,j}, \Ev^{\C A}_{w,i,j}$
 between 
$\Met^{\C A}(A_{w*i}; A_{j}) $ and $ \Met^{\C A}(A_{w};A_{i\to j})$
%
form a family of retractions (resp.~isomorphisms) natural in $w$ and non-expansive (with respect to the distance functions $\Xi_{d^{\C A}_{w},d^{\C A}_{j}}$ over $\Met^{\C A}(A_{w}; A_{j})$).

\end{definition}

The definition above can be expressed in more abstract terms using the language of \emph{enriched} categories:
the multicategory $\Met^{\C A}$ is enriched over the cartesian closed category of \emph{pre}-metric spaces and non-expansive functions (see the discussion of $\Xi$ in Section \ref{section1plus}), where $\Met^{\C A}(\vec X ;Y)$ is endowed with the pre-metric $\Xi_{\max\{\vec a\},b}$. Then 
 $\C A $ is a {quantitative (resp.~extensional) $\lambda$-algebra} when the 
maps $\LL^{\C A}_{w*i,j}, \Ev^{\C A}_{w,i,j}$ form an enriched natural retraction (resp.~isomorphism) from $\Met^{\C A}(A_{w*i};A_{j})$ to $\Met^{\C A}(A_{w};A_{i\to j})$ (notice that this implies that the pre-metrics $\Xi_{\max\{\vec a\},b}$ are indeed metrics).

\begin{example}\label{ex:reals6}
$\C I$ becomes a quantitative $\lambda$-algebra by defining $\Lambda^{\C I}$ inductively on $\FreeAlgS{\lambda,\C I}{\Var}^{0}$, exploiting the cartesian closed structure of
the subcategory of $\Met$ formed by the spaces $\C I_{i}$.
%
  \end{example}

Let us now show how quantitative $\lambda$-algebras are captured by quantitative $\lambda$-theories.

\begin{definition} \label{semantics}Let  $\C A = (A, \Omega^{\C A}, \Lambda^{\C A}, d^{\C A})$ be a quantitative applicative $\lambda$-algebra, $Y=\{ y_1, \ldots, y_n\}$, $\Gamma= \{ t_k \EQUL{\epsilon_{k}}{Y}{i_k}  u_k\  |\  k=1, \ldots, m \}$. 
 $\C A $ satisfies a quantitative $\lambda$-inference $\Gamma \vdash t \EQUL{\epsilon}{X}{i} u $ (denoted $ \Gamma \vDash_{\C A}    t \EQUL{\epsilon}{X}{i} u$)
 if for all $\rho:\Var \to A$ 
 the following implication holds:\\
 \medskip
 \adjustbox{scale=0.86}{
 \hskip-0.6cm
 \begin{minipage}{1.2\textwidth}
\begin{align} 
&\text{for all } a_{1},b_1\in A_{1},  \ldots ,  a_{n},b_n \in A_{n},
\text{ such that } \delta  = \sup\{d^{\C A}(a_{j},b_{j})\mid j=1,\dots,n\}  \tag{sat$^{*}$}, \label{satstar} 
 \\
&\text{if, for all }( t_k \EQUL{\epsilon}{Y}{i_k} u_k)\in \Gamma , \text{ we have } d^{\C A}_{i_k} ( f_{\vec{y},\vec{a}}^{\natural}(t_{k}), f_{\vec{y},\vec{b}}^{\natural}(u_{k}))  \leq \max\{\delta,\epsilon \}
\text{\mbox{ then } }  
  d^{\C A}_i ( f_{\vec{y},\vec{a}}^{\natural}(t)  , f_{\vec{y},\vec{b}}^{\natural}(u))\leq\max\{\delta, \epsilon\}   
  \notag
      \end{align}
  \end{minipage}
  } \\
   \medskip
$\C A $ satisfies a quantitative  $\lambda$-theory $\C U$ (denoted $ {\C A} \vDash \C U$) if it satisfies all the inferences in $\C U$.
\end{definition} 

Condition \eqref{satstar} from Def.~\ref{semantics} is admittedly more complex than condition \eqref{sat}  from Def.~\ref{def:sat}.
Yet, this is the price one has to pay in order to be able to express the quantitative $\xi$-rule. Indeed, condition \eqref{satstar} treats ``locally quantified'' variables by applying a condition reminiscent of the metrics $\Xi$ from Section \ref{section1plus}: for all locally quantified variables $\vec x$ in $t\EQUL{\epsilon}{X}{i} u$, when the $\vec x$ are replaced in $t$ and $u$ by \emph{different} points $\vec a,\vec b$, the distance between the resulting terms must be bounded 
by either $\epsilon$ or the max of the $d^{\C A}(a_{j},b_{j})$.
This ensures that, whenever $\C A\vDash t \EQUL{\epsilon}{\{x\}}{j} u$ is satisfied, 
we can conclude $\Xi(\lambda x.t,\lambda x.u)\leq \epsilon$, as the $\xi$-rule requires.

\begin{example}
Def.~\ref{semantics} solves the problem from Example \ref{ex:reals4}: with $r=0$ and $s=\epsilon$, from the fact
that $| f(s)- g(r)| =2\epsilon > \max\{ |r- s|, \epsilon\}$, it follows 
that $\C I\not\vDash \overline fx \EQUL{\epsilon}{X}{j} \overline gx$, hence blocking the counter-example to the $\xi$-rule. 
Rather, it holds that $\C I\vDash\C  U_{\lambda\eta}^{\C I}$ 
(cf.~Example \ref{ex:reals5}).

\end{example}

\begin{proposition}

A quantitative applicative $\lambda$-algebra is a quantitative $\lambda$-algebra (resp.~a quantitative extensional $\lambda$-algebra) iff it satisfies $U_{\lambda}$ (resp.~$U_{\lambda\eta}$).
\end{proposition}
\longversion{
\begin{proof}
The only-if direction is easily checked. For the if-direction, it suffices to check that the validity of the rule $(\xi)$ implies that $\LL$ maps non-expansive functions onto non-expansive functions, and the validity of $(\beta)$ (resp.~$(\beta)$ and $(\eta)$) ensures that $(\LL, \Ev)$ yield a retraction (resp.~an isomorphism).
\end{proof}
}

From the argument of the proposition above one can also deduce that a
quantitative applicative $\lambda$-algebra is
a weak $\lambda$-algebra iff it satisfies ($\alpha$) and ($\beta$), and is an (extensional) $\lambda$-algebra iff it furthermore satisfies $(\xi)$ (and $(\eta)$).

We conclude this section by showing soundness and completeness of quantitative (extensional) $\lambda$-theories. \shortversion{The proof is based on the construction of a ``quantitative term model''.}

%


\begin{theorem}[Soundness and Completeness of Quantitative $\lambda$-theories]\label{thm:completeness2}
Let $\C U$ be a quantitative $\lambda$-theory (resp.~a quantitative extensional $\lambda$-theory) over $\Sigma^{\lambda}$. Then $\Gamma \vLambda \phi \in \C U$ (resp.~$\Gamma \vLambdaEta \phi\in \C U$) iff 
$\Gamma \vDash_{\C A}\phi$ holds for any quantitative $\Sigma^{\lambda}$-algebra (resp.~quantitative extensional $\Sigma^{\lambda}$-algebra) $\C A$ such that $\C A\vDash \C U$.
\end{theorem} 
\longversion{

Soundness is checked by induction on the rules. 
The rest of this section is devoted to proving the completeness part of this theorem.
This will be obtained by constructing a suitable ``term model'' of the smallest $\lambda$-theory $\VLambda$ containing $\Gamma$ and $\C U$
(resp.~of the smallest extensional $\lambda$-theory $\VLambdaEta$ containing $\Gamma$ and $\C U$).
 For any $X\subseteq_{\mathrm{fin}}\Var$ and 
 sort $i$, let $\dLambda{\VLambda}_{X,i}$ be the distance on $\FreeAlgS{\lambda}{\Var}_{i} $ defined by 
$
\dLambda{\VLambda}_{i}(t,s)= \inf\{ \epsilon\in \BB Q^{+} \mid \emptyset \vLambda t \EQUL{\epsilon}{X}{i}s \in \VLambda\}
$, and similarly  let $\dLambdaEta{\VLambda}_{X,i}(t,s) :=\inf\{ \epsilon\in \BB Q^{+} \mid \emptyset \vLambdaEta t \EQUL{\epsilon}{X}{i} s \in \VLambdaEta\}
$
. 
One can check that $\dLambda{\VLambda}_{X,i}$ and $\dLambdaEta{\VLambda}_{X,i}$ are indeed pseudo-metrics. Moreover, one has the following:

\begin{proposition}\label{lemma:termdistance}
For all sorts $i,j\in I$, and terms $t,s\in \FreeAlgS{\lambda}{\Var}_{j}$, \\
\adjustbox{scale=0.85}{
\begin{minipage}{1.2\textwidth}
\begin{align}
\dLambda{\VLambda}_{X,i\to j}(\lambda_{i}x.t, \lambda_{i}x.s) & = \inf\{\delta \mid 
\forall v,w\in \FreeAlgS{\lambda}{\Var}_{X,i} \ 
\dLambda{\VLambda}_{X\cup\{x:i\},j}(t[v/x],s[w/x]) \leq \sup\{\dLambda{\VLambda}_{X\cup\{x:i\},i}(v,w), \delta\}\}
\tag{$\dag$}\label{dag}
\end{align}
\end{minipage}
}
\medskip

and for all terms $t,s\in \FreeAlgS{\lambda}{\Var}_{i\to j}$,
and variable $x$ not occurring free in either $t$ or $s$,\\
\adjustbox{scale=0.85}{
\begin{minipage}{1.2\textwidth}
\begin{align}
\dLambdaEta{\VLambda}_{X,i\to j}(t,s) & = \inf\{\delta \mid 
\forall v,w\in \FreeAlgS{\lambda}{\Var}_{i} \ 
\dLambdaEta{\VLambda}_{X\cup\{x:i\},j}(tv,sw) \leq \sup\{\dLambdaEta{\VLambda}_{X\cup\{x:i\},i}(v,w), \delta\}\}
\tag{$\ddag^{\boldsymbol\eta}$}\label{ddageta}
\\
\dLambda{\VLambda}_{X,i\to j}(t,s) & \geq \inf\{\delta \mid 
\forall v,w\in \FreeAlgS{\lambda}{\Var}_{i} \ 
\dLambda{\VLambda}_{X, j}(tv,sw) \leq \sup\{\dLambda{\C V}_{X,i}(v,w), \delta\}\}
\tag{$\ddag$}\label{ddag}
\end{align}
\end{minipage}
}
\end{proposition}

To establish Proposition \ref{lemma:termdistance} we need a few lemmas.

\begin{lemma}\label{lemma:dist}
For all $\epsilon \in \BB R_{\geq 0}$ and $A\subseteq \BB R_{\geq0}$, $\sup\{\epsilon, \inf A\}= \inf\{\sup\{\epsilon, \delta\}\mid \delta \in A\}$.
\end{lemma}
\begin{proof}
Let $\theta=\sup\{\epsilon, \inf A\}$ and $\eta= \inf\{\sup\{\epsilon, \delta\}\mid \delta \in A\}$.
If $\epsilon \geq \inf A$, then $\theta=\eta=\epsilon$. 
If $\epsilon < \inf A$, then $\theta= \inf A$. Since $\BB R_{\geq 0}$ is totally ordered, we conclude that 
$\theta=\eta$.
\end{proof}
\begin{lemma}\label{lemma:inf}
For all $A\subseteq \BB R_{\geq 0}$ such that $A$ is upward closed, $\inf A=\inf(A\cap \BB Q)$.
\end{lemma} 
\begin{proof}
Since $A\cap \BB Q\subseteq A$ it follows that  $\inf A\leq \inf(A\cap \BB Q)$.
If $\inf A< \inf (A\cap \BB Q)$, then there exists $x\in A - \BB Q$ such that $x < \inf (A\cap \BB Q)$; by the density of $\BB Q$ in $\BB R$ there exists then $q\in \BB Q$ such that $x\leq  q< \inf (A\cap \BB Q)$, and since $A$ is upward closed, it must be $q\in A\cap \BB Q$, which is absurd.  We conclude that $\inf A \geq \inf (A\cap \BB Q)$.
\end{proof}
\begin{lemma}\label{lemma:inf2}
Let $\epsilon \in \BB Q_{\geq 0}$ and $A\subseteq \BB Q_{\geq 0}$. If $\epsilon \geq \inf A$ then either $\epsilon= \inf A$ or there exists $\delta \in A$ such that $\epsilon \geq \delta$.
\end{lemma}
\begin{proof}
Suppose $\epsilon \geq \inf A$.
If $\epsilon < \delta$ holds for all $\delta \in A$, then $\epsilon \leq \inf A$, and this is possible only if $\epsilon=\inf A$.
\end{proof}
\begin{lemma}\label{lemma:arch}
For all $\epsilon \in \BB Q_{\geq 0}$, if $\epsilon \geq \dLambda{\VLambda}_{X,i}(t,s)$ (resp.~$\epsilon \geq \dLambdaEta{\VLambda}_{X,i}(t,s)$), then $\emptyset \vLambda t \EQUL{\epsilon}{X}{i}s\in \VLambda$
(resp.~$\emptyset \vLambdaEta t \EQUL{\epsilon}{X}{i}s\in \VLambdaEta$).
\end{lemma}
\begin{proof}
Let $\delta =\dLambda{\VLambda}_{X,i}(t,s)$.
If $\epsilon > \delta$, then by Lemma \ref{lemma:inf2} there exists a rational number $\epsilon'\leq \epsilon$ such that 
$\emptyset \vLambda t\EQUL{\epsilon'}{X}{i} s$, and we deduce
$\emptyset \vLambda t\EQUL{\epsilon}{X}{i}s$ by (Max).
If $\epsilon =\delta$, then $\epsilon =\inf\{ \eta \mid \emptyset \vLambda t\EQUL{\eta}{X}{i}s\}$, so we deduce $\emptyset \vLambda t\EQUL{\epsilon}{X}{i}s$ by (Arch).
\end{proof}

\begin{proof}[Proof of Proposition \ref{lemma:termdistance}]
We only prove \eqref{dag}, as \eqref{ddageta} and \eqref{ddag} can be proved in a similar way. Let
\begin{align*}
A& =\{\delta\in \BB Q_{\geq 0}\mid \emptyset \vLambda \lambda_{i}x. t \simeq_{\delta}^{X, i\to j}\lambda_{i}x.s\in \VLambda\}\\
B & =\{\delta \in \BB R_{\geq 0}\mid 
\forall v,w\in \FreeAlgS{\lambda}{\Var}_{i}, \ 
 \dLambda{\VLambda}_{X\cup\{x:i\},j}(t[u/x], s[v/x])\leq 
\sup\{\delta, \dLambda{\VLambda}_{X\cup\{x:i\},i}(v,w)\}\}
\end{align*}
We will first show that $A\subseteq B$, from which we deduce $\inf B \leq \inf A$.
Let then $\delta \in A$ and suppose $\emptyset \vLambda u \EQUL{\epsilon}{X\cup\{x:i\}}{i}v\in \VLambda$.
Let $Y=X\cup\{x:i\}$; 
since $x$ is not free in $\lambda_{i}x.t$ and $\lambda_{i}x.s$, using 
(Abstraction) we deduce that $\emptyset \vLambda \lambda_{i}x.t \EQUL{\epsilon}{Y}{i\to j}\lambda_{i}x.s\in \VLambda$.
 Letting $\sup\{\epsilon,\delta\}=\delta+\alpha= \epsilon+\beta$, using  (Max) we
deduce then $\emptyset \vLambda \lambda_{i}x. t \EQUL{\sup\{\epsilon,\delta\}}{Y}{i\to j}\lambda_{i}x.s\in \VLambda$ and 
$\emptyset \vLambda u \EQUL{\sup\{\epsilon,\delta\}}{Y}{i}v\in \VLambda$, so by (NExp) we deduce $\emptyset \vLambda (\lambda_{i} x.t)u \EQUL{\sup\{\epsilon,\delta\}}{Y}{j} (\lambda_{i}x.s)v\in \VLambda$, and by ($\beta$), we obtain 
 $\emptyset \vLambda t[u/x]  \EQUL{\sup\{\epsilon,\delta\}}{Y}{j} s[v/x]\in \VLambda$. Using Lemma \ref{lemma:dist}, we now have that 
\begin{align*}
\sup\{\delta, \dLambda{\VLambda}_{Y,i}(v,w)\}&  =\sup\{\delta, \inf\{\epsilon\mid \emptyset \vLambda u \EQUL{\epsilon}{Y}{i}v\in \VLambda\}\}\\
&=
\inf\{ \sup\{\epsilon, \delta\} \mid \emptyset \vLambda u \EQUL{\epsilon}{X}{i}v\in \VLambda\}\}\\
&\geq \dLambda{\VLambda}_{Y,j}(t[u/x],s[v/x])
\end{align*}
where the last step follows from the fact that $\emptyset \vLambda t[u/x] \EQUL{\sup\{\epsilon, \delta\}}{Y}{j} s[v/x]\in \VLambda$ holds, by what we have seen, for all $\epsilon$ such that $\emptyset \vLambda u \EQUL{\epsilon}{Y}{i}v\in \VLambda$.
We conclude then that $\delta \in B$.

We will now show $B\cap \BB Q_{\geq 0}\subseteq A$. Since $B$ is an upper set, this is be enough to conclude $\inf A\leq \inf B$, using Lemma \ref{lemma:inf}.
Let then $\delta $ be a rational number in $B$. Let again $Y=X\cup\{x:i\}$ and suppose $\emptyset \vLambda u \EQUL{\epsilon}{Y}{i}v$ holds for some $\epsilon \in \BB Q_{\geq 0}$; 
then 
$\sup\{\epsilon, \delta\} \geq \sup\{  \dLambda{\VLambda}_{Y,i}(u,v),\delta\} \geq \dLambda{\VLambda}_{Y,j}(t[u/x], s[v/x])= \inf\{ \alpha\in \BB Q_{\geq 0} \mid \emptyset \vLambda t[u/x]\EQUL{\alpha}{Y}{j} s[v/x]\}$.
Now, from $\emptyset \vLambda x \EQUL{0}{Y}{i} x$, we deduce 
$\delta =\sup\{0,\delta\}\geq \dLambda{\VLambda}_{Y,j}(t[x/x],s[x/x])= \inf\{ \alpha\in \BB Q_{\geq 0} \mid \emptyset \vLambda t \EQUL{\alpha}{Y}{j} s\}$.
Using Lemma \ref{lemma:arch} we deduce then $\emptyset \vLambda t  \EQU{\delta}{Y}{j}s\in \VLambda$, and by ($\xi)$ together with (Concretion) we conclude  
$\emptyset \vLambda \lambda x_{i}.t \EQUL{\delta}{X}{i\to j}\lambda_{i} x.s$, which proves that $\delta \in A$.
\end{proof}
The following lemma is an immediate consequence of ($\beta$).
\begin{lemma}\label{lemma:termbetanat}
For all sorts $i,j\in I$, and terms $t\in \FreeAlgS{\lambda}{\Var}(1,j)$, $s\in \FreeAlgS{\lambda}{\Var}(1,i)$, $u\in \FreeAlgS{\lambda}{\Var}(1,k)$, 
$d^{\mathsf e}_{X,j}((\lambda_{i}x.t)s,t[s/x]) =0 $ where $\mathsf e\in \{\boldsymbol\lambda, \boldsymbol\lambda\boldsymbol\eta\}$.
\end{lemma}

For all $X\subseteq_{\mathrm{fin}}\Var$, let $\BB P^{\boldsymbol\lambda}_{X}$ the be the set of all $I$-indexed families of metrics $p$ on $\FreeAlgS{\lambda}{\Var}$ such that for all $f: \Var\to \FreeAlgS{\lambda}{\Var}$, $f^{\sharp}$ is non-expansive, as a function from $(\FreeAlgS{\lambda}{\Var},\dLambda{\VLambda}_{X})$ to $(\FreeAlgS{\lambda}{\Var}, p)$, and let $d_{X,i}(t,s)= \sup\{p_{i}(t,s)\mid p\in \BB P^{\boldsymbol\lambda}_{X}\}$.
Similarly, let $\BB P^{\boldsymbol\lambda\boldsymbol\eta}_{X}$ the be the set of all $I$-indexed families of pseudometrics $p$ on $\FreeAlgS{\lambda}{\Var}$ such that for all $f: \Var\to \FreeAlgS{\lambda}{\Var}$, $f^{\sharp}$ is non-expansive, as a function from $\FreeAlgS{\lambda}{\Var},\dLambdaEta{\VLambda}_{X})$ to $(\FreeAlgS{\lambda}{\Var}, p)$), and let $e_{X,i}(t,s)= \sup\{p_{i}(t,s)\mid p\in \BB P^{\boldsymbol\lambda\boldsymbol\eta}_{X}\}$.
 
The lemma below show that the metrics $d_{X,i},e_{X,i}$ coincide with $\dLambda{\VLambda}_{i}$ and $\dLambdaEta{\VLambda}_{i}$, respectively.
\begin{lemma}
For all 
 $X\subseteq_{\mathrm{fin}}\Var$, 
$i\in I$ and $t,s\in \FreeAlgS{\lambda}{\Var}_{i}$, $d_{X,i}(t,s)=\dLambda{\VLambda}_{X,i}(t,s)$ and $e_{X,i}(t,s)=\dLambdaEta{\VLambda}_{X,i}(t,s)$.
\end{lemma}
\begin{proof}
For any substitution $f:\Var\to \FreeAlgS{\lambda}{\Var}$ by (Subst) we have 
that $
\emptyset \vLambda s\EQUL{\epsilon}{X}{i}t\in \VLambda$ implies  $\emptyset \vLambda f^{\flat}(s)\EQUL{\epsilon}{X}{i}f^{\flat}(t)\in \VLambda
$, 
whence $\dLambda{\VLambda}_{X,i}(t,s)\geq \dLambda{\VLambda}_{X,i}(f^{\flat}(t), f^{\flat}(s))$. This implies that $\dLambda{\VLambda}_{X}\in \BB P^{\boldsymbol\lambda}_{X}$, and thus that $d_{X,i}\geq \dLambda{\VLambda}_{X,i}$.
For the converse direction, by definition $d_{X}$ makes the function $\mathrm{id}^{\flat}: \FreeAlgS{\lambda}{\Var}\to \FreeAlgS{\lambda}{\Var}$ non-expansive, and this implies $\dLambda{\VLambda}_{X,i}(t,u) \geq d_{X,i}(t,u)$.
The argument for $e_{X,i}$ is similar.
 \end{proof}

For all $X\subseteq_{\mathrm{fin}}\Var$ and $i\in I$, let $\FreeAlgLL{X,i}$ be the set obtained by quotienting $\FreeAlgS{\lambda}{\Var}_{i}$ by 
the relation $t\simeq_{X} u$ iff $d_{X,i}(t,u)=0$. 
Let $(\FreeAlgLL{X,i}, d^{\simeq}_{X,i})$ be the metric space given by $d^{\simeq}_{X,i}(t,u)=d_{X,i}([t]_{X},[u]_{X})$. 
For any choice of $X\subseteq_{\mathrm{fin}}\Var$, we obtain then a quantitative $\lambda$-algebra $\C T_{X}^{\boldsymbol\lambda}:=(\FreeAlgLL{X}, \Omega^{\simeq}_{X}, \Lambda^{\C T} ,d^{\simeq}_{X})$, where $\Omega^{\simeq}_{X}$ contains the maps $\sigma^{\simeq}_{X}: [t_{1}]_{X},\dots, [t_{n}]_{X} \mapsto [\sigma(t_{1},\dots, t_{n})]_{X}$, and $\Lambda^{\C T}_{X,w,i}:
\FreeAlgS{\lambda, \FreeAlgLL{X}}{\Var}_{\vec w\to j}
\to \FreeAlgLL{X,\vec w\to j}
$ sends $\lambda \vec y.t$ onto $[\lambda \vec y.t]_{X}$.
Moreover, 
Lemma \ref{lemma:termdistance} and Lemma \ref{lemma:termbetanat} assure that the pair $(\Lambda^{\C T}_{X},  \lambda \vec z.t \mapsto  \lambda \vec z.\lambda_{i}x.t\vec z x)$ yields a family of retractions   
$ 
\Met^{\C T_{X}^{\boldsymbol\lambda}}(
\FreeAlgLL{X,w*i}, \FreeAlgLL{X,j})
\To
\Met^{\C T_{X}^{\boldsymbol\lambda}}(
\FreeAlgLL{X,w}, \FreeAlgLL{X,i\to j})
%
 $ natural in $w$.
 A similar argument shows the existence of quantitative extensional $\lambda$-algebras $\C T^{\boldsymbol\lambda\boldsymbol\eta}_{X} $.

To conclude the completeness argument,
fix a finite set $X$ such that all equations in $\Gamma \vdash \phi$ can be put in the form
$t\EQUL{\epsilon}{X}{i}s$. From the fact that $\C T^{\boldsymbol\lambda}_{X}$ is a quantitative $\Sigma\lambda$-algebra, and the hypothesis, it follows that $\Gamma \vDash_{\C T^{\boldsymbol\lambda}_{X}}\phi$. 
Since, by construction, $\C T^{\boldsymbol\lambda}_{X} \vDash \Gamma$, it follows that 
$\C T^{\boldsymbol\lambda}_{X}\vDash \phi$. Let $\phi= t\EQUL{\epsilon}{x}{i} s$; since $\epsilon \in \BB Q_{\geq 0}$, using Lemma \ref{lemma:arch} we deduce that $\emptyset \vLambda t \EQUL{\epsilon}{X}{i}s\in \Gamma \cup \C U$, and thus that  
$\Gamma \vLambda t \EQUL{\epsilon}{X}{i}s\in\C U$. 
A similar argument can be done for  $\C T_{X}^{\boldsymbol\lambda\boldsymbol\eta}$.
}


\begin{remark}\label{rem:partial3}
Also in this case the whole construction scales to the case of partial ultra-metric spaces.
Following Remark \ref{rem:partial2}, we will speak of
\emph{partial $\lambda$-theories} and
\emph{partial $\lambda$-algebras}.
\end{remark}

%
%
%
%
%
 
 %

\section{Metric Constraints}\label{section4}

In this section we take a closer look at the several obstacles one might face when looking for higher-order quantitative algebras.
First, as seen in Section \ref{section1plus}, in higher-order types  the \emph{unique} distance, $\Xi$, making both application and abstraction non-expansive operations might not be  a metric.
Moreover, even if such a metric exists, several conditions might lead higher-order distances to be trivial (i.e.~discrete), or have plenty of isolated points.
But discrete metrics and isolated points  convey no more information than equivalences,  while one of the main reasons to look for semantics of program distances is to be able to compare informatively programs which are not equivalent.
Despite what look like strong limitations, we conclude this section by presenting a few examples of non-discrete quantitative $\lambda$-algebras.

%
%
\subparagraph*{Existence of Exponential Objects.}

Given metric spaces $(X,a)$ and $(Y,b)$, if $(Y,b)$ is ultra-metric, then $\Xi_{a,b}=\Phi_{a,b}$ is always a metric, which means that $\Met(X,Y)$ is their exponential object in $\Met$.
When $(Y,b)$ is not ultra-metric,  condition \eqref{star} from Theorem \ref{thm:clementino} provides a useful \emph{sufficient} criterion to check if $\Xi_{a,b}$ is a metric (and thus, if some candidate quantitative applicative $\lambda$-algebra $\C A$ is a quantitative $\lambda$-algebra).
 \longversion{Indeed, \eqref{star} is equivalent (cf.~\cite{Clementino2006}, p.~3124) to the condition below
  \begin{align}
\inf_{x\in X} \{\max\{a(x_{0},x),\alpha\}+\max\{a(x,x_{2}),\beta\} \}\geq a(x_{0},x_{2})
 \tag{$**$}\label{startwo}
 \end{align}
 Now,  for any metric space $(Y,b)$, given $f,g,h\in \Met(X,Y)$, supposing $\Xi_{a,b}(f,h)+\Xi_{a,b}(h,g)\geq a(x_{0},x_{2})$, if \eqref{startwo} holds, we deduce 
 \begin{align*}
 b(f(x_{0}),g(x_{2})) &
 \leq \inf_{x\in X}b(f(x_{0}), h(x))+ b(h(x), g(x_{2})) \\
 &  \leq  \inf_{x\in X} \{ \max\{a(x_{0},x), \Xi_{a,b}(f,h)\}+\max\{a(x,x_{0}),\Xi_{a,b}(h,g)\} \}\\
 & \leq \Xi_{a,b}(f,h)+ \Xi_{a,b}(h,g)
 \end{align*}
which shows that $ \Xi_{a,b}(f,g)\leq  \Xi_{a,b}(f,h)+ \Xi_{a,b}(h,g) $. 
}
We will now show that, under very mild hypotheses, the validity of \eqref{star} is also \emph{necessary} for $\C A$ to be a quantitative $\lambda$-algebra.

 Let a quantitative applicative $\lambda$-algebra $\C A$ be \emph{observationally complete} when it 
contains the metric space $(\PReal, |\cdot-\cdot|)$ and for all sort $i$,  
$\Met^{\C A}(A_{i}; \PReal)\simeq \Met(A_{i},\PReal)$. In other words, $\C A$ contains \emph{all} observations on $A_{i}$ with target $\PReal$.
Moreover, let a \emph{quantitative $\lambda$-pre-algebra} be as a quantitative $\lambda$-algebra
 $\C A$, but where the $d^{\C A}_{i}$ need only be pre-metrics.
 Given a quantitative $\lambda$-pre-algebra $\C A$, let 
 $\C A^{*}$ indicate the restriction of $\C A$ to those sorts $i$ for which $d^{\C A}_{i}$ is a metric (i.e.~it also satisfies \trans).

\begin{proposition}\label{lemma:weakclem2}
Let $\C A$ be an observationally complete quantitative extensional $\lambda$-pre-algebra. For any $A$ in $\C A^{*}$, $A$ is exponentiable in $\Met^{\C A^{*}}$ iff 
for all $\alpha,\beta\in \mathrm{Im}(d^{\C A})$ and $x_{0},x_{2}\in X$ with $a(x_{0},x_{2})=\alpha+\beta$, condition \eqref{star} holds. 
%
%
\end{proposition}
\longversion{

 \begin{proof}
 
Let $A=A_{i}$ and $B=A_{j}\in \C A^{*}$.
Since $\C A$ is extensional, we can identify $A_{i\to j}$ with $\Met^{\C A}(A,B)$. 
First we show that for all $f,g\in A_{i\to j}$, their distance is completely determined by 
 $d^{\C A}_{i\to j}(f, g)= \Xi_{d^{\C A}_{i}, d^{\C A}_{j}}(f,g)$. 
 Let $\Ev^{*}_{i,j}:=\Ev_{i\to j,i,j}^{\C A}(\mathrm{id}_{A_{i\to j}})\in \Met^{\C A}(A_{i\to j, i},A_{ j})$
 Since $\Ev^{*}_{i,j}$ is non-expansive we have 
 $  d^{\C A}_{j}(f(a),g(b))=
 d^{\C A}_{j}(\Ev^{*}_{i,j}(f)(a),\Ev_{i,j}^{*}(g)(b))\leq \sup\{d^{\C A}_{i}(a,b),d^{\C A}_{i\to j}(f,g)\}$, which proves that 
 $d^{\C A}_{i\to j}(f,g)\geq \Xi_{d^{\C A}_{i},d^{\C A}_{j}}(f,g)$. 
Conversely, since $\LL_{i,j}^{\C A}: \Met^{\C A}(A_{w*i},A_{j})\to  \Met^{\C A}(A_{w},A_{i\to j})$ is non-expansive, we have that 
$\Xi_{d^{\C A}_{i},d^{\C A}_{j}}(f,g)\geq d^{\C A}_{i\to j}(\LL^{\C A}_{i,j}(f),\LL^{\C A}_{i,j}(g))$.

Now, if condition \eqref{star} holds, we can argue as above to show that the distance functions $d^{\C A}_{i\to j}=\Xi_{d^{\C A}_{i},d^{\C A}_{j}}$ are all metrics.
To show that condition \eqref{star} is necessary, one can now argue as in \cite{Clementino2006} [pp.~3126-3127] (using the fact that $\Met^{\C A}(A;\PReal)=\Met(A,\PReal)$), showing that if $\Xi_{d^{\C A}_{i},d^{\C A}_{j}}$ is a metric, then condition
\eqref{startwo} must hold. 
 \end{proof}

}
Proposition \ref{lemma:weakclem2} has a positive side: it provides a sufficient condition for exponentiability which is slightly weaker than Theorem \ref{thm:clementino}, as \eqref{star} needs only hold for distances $\alpha,\beta$ in the \emph{image} of the distance functions $d^{\C A}$ of the pre-algebra. 
Notice that, if the $d^{\C A}$ are discrete, condition \eqref{star} trivially holds.
On the other hand, Proposition \ref{lemma:weakclem2} has a negative side: 
if condition \eqref{star} fails (i.e.~some space $A$ does not contain ``enough points''), then $\C A$ fails to be a quantitative $\lambda$-algebra.
For instance, no algebra containing $\BB N$, with the metric inherited from $\BB R$, as one of its objects, can be a $\lambda$-algebra.

\subparagraph*{Existence of Compact Algebras.}

We have the following negative result. 
\begin{proposition} \hfill \label{counterexample}
There are no non-trivial one-sorted weak quantitative $\lambda$-algebras in Met which are compact.
 \end{proposition}
\longversion{
\begin{proof} 
Assume the converse. 
 Since the spaces are compact, application is uniformly continuous. Let ${\KKK^n\III}=\lambda\underbrace{x_1\ \ldots\ x_n}_n.{\III}$, \ie\ ${\KKK^{n+1}\III}={\KKK}({\KKK^n\III})$. Consider the sequence $\{{\KKK^n\III}\}_{n}$, then by compactness there exists a converging subsequence $\{{\KKK^{n_i}\III}\}_{i}\rightarrow \overline{x}$.  Now, since application is uniformly continuous, for each $\frac{1}{k}$ consider the corresponding  ${\KKK^{n_k}\III}$ such that for all $m\geq n_k$, $ d(y\cdot{\KKK^{m}\III},y\cdot{\KKK^{n_k}\III})\leq \frac{1}{k}$, for all $y$. Take $y= \lambda x. x\underbrace{\III\ldots \III}_{n_k}$, then by uniform continuity  of $\cdot$ we have  that   $ d({\KKK^{n_{k+1}-n_k}\III},{\III})\leq \frac{1}{k}$   for all $k$. Hence we can define a sequence $\{{\KKK^{n_{k+1}-n_k}I}\}_{i}\rightarrow {\III}$.  Now consider any normal form $t$, by continuity,  we get that the sequence $\{{\KKK^{n_{k+1}-n_k -1}\III}\}_{i}\rightarrow \p{t}$ for all $t$. Contradiction.
\end{proof}
}

By contrast, in the multi-sorted case, compact $\lambda$-algebras do exist, e.g.~take the restriction of the quantitative $\lambda$-algebra $\C I$ to compact intervals $[a,b]$, i.e.~with $a,b<\infty$, or simply the full type structure on a finite base set. 

%
%



\subparagraph*{Distances and Observational Equivalence.}

The next two results relate distances in quantitative $\lambda$-theories with observational equivalence for the associated $\lambda$-theory, clearly indicating 
the (limited) extent to which a metric can deviate from being discrete on {\em pure} closed $\lambda$-terms. We recall that {\em pure} means that no constants appear in the syntax, or categorically, that the $\Sigma^{\lambda}$-signature has only the $\cdot$ symbols. %

\begin{proposition}\label{prop:obs1}
In a quantitative $\lambda$-algebra $A$, \ie\ a model of the simply typed $\lambda$-calculus, terms which are not equated in the maximal theory are all at the same distance from one another. Moreover each $A_i$ is a bounded pseudo-metric space. 
\end{proposition}
\longversion{
\begin{proof}
The maximal non-trivial theory of the simply typed $\lambda$-calculus is the theory of the \emph{full type structure} over a two-element basic set, $A_o$, for $o$ a basic sort.  Alternately two closed terms $t$ and $s$, of sort $i :=i_1\rightarrow \ldots\rightarrow i_n\rightarrow o$, are equated in the maximal theory $\cal T$, if and only if there does not exist a term $u$ such that $ut\rightarrow_\beta \KKK^{o\rightarrow o \rightarrow o}$, $us\rightarrow_\beta \OOO^{o\rightarrow o \rightarrow o}$, and  $\KKK^{o\rightarrow o \rightarrow o}:=\lambda x^o.\lambda y^o.x$ and  $\OOO^{o\rightarrow o \rightarrow o}:=\lambda x^o.\lambda y^o.x$, where $\rightarrow_{\beta\eta}$ denotes $\beta\eta$-conversion, see \cite{Baren13}. 
Suppose that $t\neq_{\cal T}s$. Then there exists $u$ such that $ut=_\beta \KKK^{o\rightarrow o \rightarrow o }$ and $us=_\beta \OOO^{o\rightarrow o \rightarrow o}$. Then $d(\p{ut},p{us})=d(\p{\KKK},\p{\OOO})\leq d(\p{t},\p{s})$. To see the converse consider 
$$u':=\lambda z^o\ \lambda z_1^{i_1}\ldots \lambda z_n^{i_n}z(tz_1\ldots z_n)(sz_1\ldots z_n)$$ then  we have $d(\p{u'\KKK},\p{u'\OOO})=d(\p{t},\p{s})\leq d(\p{\KKK},\p{\OOO})$.

\noindent Since we  have that $\p{\KKK^{i\rightarrow i \rightarrow i}ts}=\p{t}$ and $\p{\OOO^{i\rightarrow i \rightarrow i}ts}=\p{s}$, we have also that $d(\p{\KKK^{i\rightarrow i \rightarrow i}ts},\p{\OOO^{i\rightarrow i \rightarrow i}ts})=d(\p{t},\p{s})\leq d(\p{\KKK^{i\rightarrow i \rightarrow i}},\p{\OOO^{i\rightarrow i \rightarrow i}})$. Since $\p{\KKK^{i\rightarrow i \rightarrow i}}$ and $\p{\OOO^{i\rightarrow i \rightarrow i}}$ are not equated in the maximal theory, as can be easily seen by taking 
$$ \lambda z^{i\rightarrow i \rightarrow i}\lambda x^o.\lambda y^o.z(\lambda z_1^{i_1}\ldots\lambda z_n^{i_n}.x )(\lambda z_1^{i_1}ldots\lambda z_n^{i_n}.y ) z_1^{i_1}ldot z_n^{i_n}$$ 
we have that the $A_i$ are all bounded by $d(\p{\KKK^{o\rightarrow o \rightarrow o}}, \p{\OOO^{o\rightarrow o \rightarrow o}})$.
\end{proof} 
}
The maximal non-trivial theory of the \emph{pure} simply typed $\lambda$-calculus is the theory $FTS$ of the \emph{full type structure} over a two-element base set (\cite{Baren13}). 
Proposition \ref{prop:obs1} implies then that  
any quantitative $\lambda$-algebra for $FTS$ is discrete. We recall that ``pure'' means that no constants appear in the syntax.
Next, we consider the untyped $\lambda$-calculus:
\begin{proposition}\label{prop:obs2}
In a non-trivial weak quantitative $\lambda$-algebra $A$, the maximal distance between any two points is 
bounded by $d(\p{\KKK},\p{\KKK(\SSS\KKK\KKK)})$. Hence all pairs of terms which can be applied, by a given term, on $\p{\KKK}$ and $\p{\KKK(\SSS\KKK\KKK)}$ respectively, are that distance apart. Moreover, if  $A$ is a non-trivial quantitative $\lambda$-algebra then for any two  solvable terms,  $t$ and $s$, which are not equated in the maximal theory ${\cal H}^*$ (see \cite{Baren85}) and $\YYY$, fixed-point combinator we have  $d(\p{t},\YYY\KKK)=d(\p{s},\YYY\KKK)$. If the distance is ultrametric we have also $d(\p{t},\YYY\KKK)= d(\p{t},\p{s})$. In any case, if the theory equates all unsolvable terms then $d(\p{t},\YYY\KKK)\leq d(\p{t},\p{s})$. 
\end{proposition}
\longversion{

\begin{proof}
We immediately have that for all terms $t$ and $s$, $d(\p{\KKK ts},\p{\KKK(\SSS\KKK\KKK) ts})=d(\p{t},\p{s})\leq d(\p{\KKK},\p{\KKK(\SSS\KKK\KKK)})$ so $A_\star$ is bounded. Moreover if there exists $u$ such that  $ut=\KKK$ and $us=\KKK(\SSS\KKK\KKK)$ then $d(\p{ut},\p{us})=d(\p{\KKK},\p{\KKK(\SSS\KKK\KKK)})\leq d(\p{t},\p{s})$

\noindent If $t\neq_{{\cal H}^*}s$ then there exists a term $u$ such that $ut$ is solvable but $us$ is unsolvable, by definition of the maximal theory ${\cal H}^*$, see \cite{Baren85}. By defintion of solvability we may assume that, say $t$ is solvable and $ut=_{\beta\eta}\III$. Therefore $d(\p{ut},\p{us})= d(\p{\III},\p{us})\leq d(\p{t},\p{s})$.  In the head reduction starting from $ux$, the variable $x$  has to appear, necessarily, as head variable. So we have also that if $t$ is a solvable term then $d(\p{ut},\p{u(\YYY\KKK)})= d(\p{\III},\p{\YYY\KKK})\leq d(\p{t},\p{\YYY\KKK})$. But clearly we have also  that $d(\p{\III t},\p{(\YYY\KKK)t})= d(\p{t},\p{\YYY\KKK})\leq d(\p{\III},\p{\YYY\KKK})$. So all solvable terms are at the same distance from$\p{\YYY\KKK}$. If all unsolvables are equated than from the first inequality we get $d(\p{\III},\p{\YYY\KKK})\leq d(\p{t},\p{s})$.
\end{proof}
}

As a consequence of B\"ohm Theorem (see \cite{Baren85}),  
Proposition \ref{prop:obs2} implies that 
any quantitative $\lambda$-algebra for the \emph{pure} untyped $\lambda$-calculus is discrete over $\beta\eta$-normal forms.

\subparagraph*{Positive Examples.}

The above limiting results apply only to terms which are not equated in the maximal theories of the $\lambda$-calculus, either typed or untyped (\cite{Baren85,Baren13}). Clearly these terms are significant computationally, and this is the bad news, but these terms are rather special and hence Propositions \ref{prop:obs1} and \ref{prop:obs2} have only a limited negative impact, and this is the good news. 
For instance, in the maximal theory of simply typed $\lambda$-calculus Church, numerals are equated up to parity, 
so Proposition \ref{prop:obs1}  does not have any bearing on the mutual distance of two different even, 
or two different odd, Church Numerals. There exist indeed rather intriguing distances in quantitative $\lambda$-algebras, even in the category of complete (not necessarily ultra-) metric spaces and non-expansive functions, as the following examples show.

 \noindent Any complete partial order model of Combinatory Logic, and hence in particular of $\lambda$-calculus (\eg\ any Scott's inverse limit  D$_\infty$ model, \cite{Baren85}), can be endowed with the metric \\
$d(d_1,d_2) 
= \begin{cases}
      0 & \text{if}\  d_1=d_2\\
            1/2& \text{if}\ d_1 \text{ and } d_2 \text{ have an upper bound }\\
      1 & \text{otherwise .}
    \end{cases}$
\\ One can check that application is non-expansive, and that the space is complete; moreover the space of representable functions (\ie\ functions determined by the elements of the model), endowed with the supremum metric, is isometrically embedded in the space. \\
Alternatively, one can consider the  term model of the simply typed $\lambda$-calculus with a base constant $\perp$. By strong normalization, it consists of the $\beta\eta$-normal forms. Let $\sqsubseteq $ be the order relation defined on normal forms of the same type by $\lambda \vec{x}. \perp \sqsubseteq \lambda \vec{x}. t$ and 
  $\lambda \vec{x}. x_i t_1 \ldots t_k \sqsubseteq \lambda \vec{x}. x_i t'_1 \ldots t'_k$, if $t_i \sqsubseteq t'_i$ for all $i=1, \ldots , k$ (corresponding to the natural order relation on B\"ohm trees, see \cite{Baren85}).
The set of $\beta\eta$-normal forms can be endowed with a notion of distance by putting, for all type $\sigma\in T$ and $t,t'$ terms of type $\sigma$, $d_{\sigma} (t,t')$ be $0$ if $t=t'$, $1/2$ if $t$ and $t'$ have an upper bound, and $1$ otherwise.
%

 Yet other distances can be given on the term model of the simply typed $\lambda$-calculus by putting  
$d_{\sigma} (t,t')$ be $0$ if $t= _{\beta\eta}t'$ and otherwise $1/N$, where $N=\max\{\ n\mid \p{t}=\p{t'}\text{ in the full type hierarchy over $n$ points }\}$.

\section{Partial Quantitative $\lambda$-Algebras}\label{section5}\label{section5}

In this section we discuss \emph{partial metrics}, and the natural generalization of quantitative $\lambda$-algebras to \emph{partial quantitative $\lambda$-algebras}. In particular we define two non-trivial such algebras for the simply typed $\lambda$-calculus. The first $\lambda$-algebra that we consider is defined on the term model of $\beta\eta$-normal forms of the simply typed $\lambda$-calculus  with a constant $\perp$ of base type.
The latter is defined within a $D_\infty$ $\lambda$-model {\em \`a la} Scott. In both cases we define an ultrametric  distance using a suitable notion of {\em term approximants}. \medskip

\noindent {\bf The Partial  $\lambda$-Algebra of the Term Model.} 
Let $T$ be the set of simple types built over the base type $o$, and let $\sigma,\tau$ range over  $T$.
The $\beta\eta$-normal forms of the simply typed $\lambda$-calculus with  constant $\perp$ of type $o$  can be endowed with a structure of applicative $\lambda$-algebra:

\begin{proposition}
 Let ${\cal NF} = ( \mathit{NF},  \Omega^{\mathit{NF}}, \Lambda^{\mathit{NF}} )$ be the structure where 
 \\  -- $\mathit{NF}$ is the $T$-indexed set of $\beta\eta$-normal forms of   typed $\lambda$-calculus with constant $\perp$ of type $o$,
  \\ --  for all $\sigma, \tau$, $\cdot_{\sigma, \tau}: \mathit{NF}_{\sigma\to \tau} \times \mathit{NF}_{\sigma} \to \mathit{NF}_{\tau}$ is defined by   
   $t \cdot_{\sigma, \tau} s = [ts]_{\beta\eta}$,
  \\
-- $\Lambda^{\mathit{NF}}_{\sigma,\tau}:  {\bf \Lambda}^0_{\sigma\rightarrow \tau}\to  \mathit{NF}_{\sigma\to \tau}$ is defined by: 
 $\Lambda^{\mathit{NF}}_{\sigma, \tau} (\lambda x.t ) = [\lambda x.t ]_{\beta\eta}$,

\noindent where  $ [t]_{\beta\eta}$ denotes the $\beta\eta$-normal form of $t$.
Then ${\cal NF}$ is an applicative $\lambda$-algebra.
\end{proposition}

The signature of this algebra can be enriched with  {\em projection operators} providing the approximants of a given normal form. Intuitively, the $n$\textsuperscript{th} approximant of a normal form is the term whose B\"ohm tree (\cite{Baren85})is obtained by cutting all branches at depth $n$ and by  labelling  leaves at level $n$ of type $\sigma_1 \rightarrow \ldots \rightarrow \sigma_m \rightarrow o$ by the term  $\lambda x_1 \ldots x_m . \perp$. More precisely:

 \begin{definition}
  For all $\sigma\in T$ and for all $n\in {\mathbb N}$, $\pi^n_{\sigma} :   \mathit{NF}_{\sigma} \to \mathit{NF}_{\sigma}$ is defined by induction on $n$  as follows:
  for all $ \lambda x_1 \ldots x_m. x_i t_1 \ldots t_k \in  \mathit{NF}_{\sigma}$,  
  \\  \hspace*{2.2cm} $\pi_0^{\sigma} (t) = \lambda   x_1 \ldots x_m. \bot \ \ \ \  \pi_{n+1}^{\sigma} (t) =  \lambda x_1 \ldots x_m. x_i (t_1)_n \ldots (t_k)_n \ .$
\end{definition} 

In the sequel, we will  denote the approximant $\pi_{\sigma}^n(t)$ simply by $t_n$.

\begin{definition}[Distance on Normal Forms] \label{dnf}
We define a family of functions $d^\mathit{NF}= \{d^\mathit{NF}_{\sigma}\}_{\sigma}$, where  $d^\mathit{NF}_{\sigma}: \mathit{NF}_{\sigma} \times \mathit{NF}_{\sigma} \rightarrow {\mathbb R}_{\geq 0}$ is defined inductively by 
%
%
\\  \hspace*{3.6cm} $d^{\mathit{NF}}_{o}(t, s) =\begin{cases}  0 & \mbox{ if } t=s
 \\ 1 & \mbox{ otherwise }  \end{cases} $\ \ \  and  \ \ \ $d^\mathit{NF}_{\sigma\rightarrow \tau} (t,t') =
\frac{1}{2^m} \ , $
\\ 
 where $m$ is the largest $n\in {\mathbb N}$, if it exists, such that  
\\ \hspace*{0.7cm} {\bf(1)}  $t_n=t'_n$  \ \ and \ \  {\bf (2)}   $\forall s,s' \in \mathit{NF}_{\sigma}.\ (d^\mathit{NF}_{\sigma} (s,s') \leq \frac{1}{2^n} \ \Longrightarrow \ d^\mathit{NF}_{\tau} (ts,ts') \leq \frac{1}{2^n} )$ ;
\\
if such a maximal  $n$ does not exist, then $d^\mathit{NF}_{\sigma\rightarrow \tau}(t,t')$ is set to $0$.

\end{definition} 

\begin{lemma}\label{dnfpm}
For all $\sigma,\tau$, $( \mathit{NF}_{\sigma}, d_{\sigma}^{\mathit{NF}})$ is a partial ultrametric space and $\cdot_{\sigma,\tau}$ is non-expansive.
\end{lemma}
\longversion{
\begin{proof} The fact that  $\cdot_{\sigma,\tau}$ is non-expansive follows immediately from the definition of $d_{\sigma}^{\mathit{NF}}$.
\\ Conditions for $d^{\mathit{NF}}$ to be an ultrametric are shown by induction on types. Symmetry is immediate. Now we  show, by induction on types that, for all $\sigma$, 
\begin{equation} \forall t,t',t'' ( \dnf{\sigma}(t,t) \leq \dnf{\sigma} (t,t') \ \mbox{and}\ \dnf{\sigma}(t,t') \leq \mbox{max} \{ \dnf{\sigma}(t,t''), \dnf{\sigma}(t'', t') \} ) \ . \label{aaa} \end{equation}
\noindent For the base type $o$ the thesis is immediate, since the only normal forms of type $o$ are $\bot$ and the variables. Then let $\sigma$ be  $\sigma_1\to \sigma_2$. We show that the first conjunct of~\ref{aaa} holds. Let $\dnf{\sigma_1\to\sigma_2}(t,t') \leq\frac{1}{2^n} $.  We prove that
 $\dnf{\sigma_1\to\sigma_2}(t,t) \leq\frac{1}{2^n} $. Condition 1 in Definition~\ref{dnf} is trivially verified. We prove condition 2.
  Let  $\dnf{\sigma_1}(s,s') \leq\frac{1}{2^n} $.  
  By the hypothesis    $\dnf{\sigma_1\to\sigma_2}(t,t') \leq\frac{1}{2^n} $,  we have $\dnf{\sigma_2} (ts, t's')\leq \frac{1}{2^n} $. By induction hypothesis, $\dnf{\sigma_1} (s', s')\leq \frac{1}{2^n} $, hence  $\dnf{\sigma}(t's', ts') \leq \frac{1}{2^n}$. By induction hypothesis, $\dnf{\sigma_2} (ts, ts')\leq \mbox{max}\{ \dnf{\sigma_2}(ts,t's'), \dnf{\sigma}(t's', ts') \}$, hence   $\dnf{\sigma_2} (ts, ts')\leq \frac{1}{2^n} $.
We are left to show the second conjunct of~\ref{aaa}, \ie\  that $\dnf{\sigma_1 \to \sigma_2 }(t,t') \leq \mbox{max} \{ \dnf{\sigma_1\to \sigma_2}(t,t''), \dnf{\sigma_1\to\sigma_2}(t'', t') \} )$.
Assume $\dnf{\sigma_1 \to \sigma_2} (t,t'') \leq \frac{1}{2^k}$ and $ \dnf{\sigma_1 \to\sigma_2} (t'',t')\leq \frac{1}{2^m}$, with $m\geq k$. We show that $\dnf{\sigma_1 \to \sigma_2} (t,t') \leq \frac{1}{2^k}$. Condition  1 of Definition~\ref{dnf}  is straightforward. We show condition 2.  Let   $\dnf{\sigma_1 } (s,s') \leq \frac{1}{2^k}$. 
Since by hypothesis we have $ \dnf{\sigma_1 \to\sigma_2} (t,t'')\leq \frac{1}{2^k}$ then $\dnf{ \sigma_2} (ts,t''s') \leq \frac{1}{2^k}$. But by hypothesis we have also $ \dnf{\sigma_1 \to\sigma_2} (t'',t')\leq \frac{1}{2^m}\leq \frac{1}{2^k}$, then, using 
  reflexivity, by induction hypothesis, \ie\ $\dnf{\sigma_1 } (s',s') \leq \frac{1}{2^k}$, we have $\dnf{ \sigma_2} (t''s', t's') \leq \frac{1}{2^k}$. Again by induction hypothesis, we have $\dnf{\sigma_2}(ts,t's') \leq \mbox{max} \{ \dnf{\sigma_2}(ts,t''s'), \dnf{\sigma_2}(t''s', t's') \} )$, and hence $\dnf{\sigma_2}(ts,t's') \leq \frac{1}{2^k}$.
 \end{proof}}
 
 \shortversion{The fact that  $\cdot_{\sigma,\tau}$ is non-expansive follows immediately from the definition of $d_{\sigma}^{\mathit{NF}}$.}
 
\begin{remark} Notice that $d^{\mathit{NF}}$ is not reflexive.
{\em E.g.},  let $u= \lambda x.xI$ of appropriate type $\sigma_1 \to \sigma_2$, $t= \lambda x.x (xt')$ and  $s= \lambda x.x (xs')$ of type $\sigma_1$, and $t',s'$ of appropriate type $\tau$ such that  $\dnf{\tau} (t',s')=1$.
 Then  $\dnf{\sigma_1} (t,s)=\frac{1}{2}$, but $d_{\sigma_2} (ut,us)=1$, \ie\ $d_{\sigma_{1}\to\sigma_{2}}(u,u)=1$.
\end{remark}

From the definition of  $ d^{\mathit{NF}}$, it immediately follows that application on normal forms is a non-expansive operator. Hence we have:

\begin{proposition}
 ${\cal NF} = ( \mathit{NF},  \Omega^{\mathit{NF}}, \Lambda^{\mathit{NF}} ,  d^{\mathit{NF}})$ is a partial quantitative extensional $\lambda$-algebra.
\end{proposition}

\begin{remark} \label{remn}
If we drop condition 2 in Definition~\ref{dnf}, then we get an ultrametric  (reflexivity holds), however application is expansive. Namely, let $t= \lambda x_1 x_2. x_1 (x_2 t')$ and
$s=  \lambda x_1 x_2. x_1 (x_2 s')$ of appropriate types $\sigma_1$ such that  $\dnf{\sigma'}(t',s') =1$. Then $\dnf{\sigma_1}(t,s) = 1/4$, but for $u = \lambda x.x II$ of type $\sigma_1\to\sigma_2$, we get $\dnf{\sigma_2}(ut, us) =1$. Notice that the above terms can be taken to be affine (similar counterexamples can be built also in the purely linear case).
\end{remark}
\medskip

\noindent {\bf The Partial  $\lambda$-Algebra of $D_{\infty}$.}
Any inverse limit domain model {\em \`a la} Scott of $\lambda$-calculus yields an applicative $\lambda$-algebra. On such models a notion of  approximant naturally arises, by considering for any given element of the domain its projections on the domains of the inverse limit construction. This leads to the following definition:

\begin{definition} \label{dinf}
\hfill
\\ (i) Let $D_{\infty}= \bigsqcup_n D_n$ be an inverse limit domain model {\em \`a la} Scott. For all $a \in D_{\infty} $ we define the $n$\textsuperscript{th} approximant of
$a$, $a_n$, as the projection of $a$  into $D_n$. 
\\ (ii) Let $D$ be the $T$-indexed set $\{ D_{\sigma}\}_{\sigma}$, where, for all $\sigma$, $D_{\sigma} = {D}_{\infty}$.
\\ (iii) Let $\dd{\sigma} : D_{\sigma} \times D_{\sigma} \rightarrow {\mathbb R}_{\geq 0}$ be the distance function defined by induction on types by
\\  \hspace*{2.7cm}  $\dd{o}(a, b) =\begin{cases}  0 & \mbox{ if } a=b
 \\ 1 & \mbox{ otherwise }  \end{cases} $
\ \ \  \  and \ \ \ \ $\dd{\sigma\rightarrow \tau} (a,b) =
\frac{1}{2^m} $ ,
\\
 where $m$ is the maximal $n\in {\mathbb N}$, if it exists, such that 
\\   \hspace*{0.6cm} {\bf (1)} 
$a_n=b_n$  \ \ \ \ and\ \ \ \   {\bf (2)} $\forall c,d \in D_{\sigma}.\ (\dd{\sigma} (c,d) \leq \frac{1}{2^n} \ \Longrightarrow \ \dd{\tau} (ac,bd) \leq \frac{1}{2^n} )$ ;
\\
if such a maximal $n$ does not exists, then $\dd{\sigma \rightarrow \tau}(a,b)=0$.
\end{definition}

\longversion{The proof of the following lemma is similar to that of Lemma~\ref{dnfpm}. }

\begin{lemma}
For all $\sigma, \tau$, $(D_{\sigma} , \dd{\sigma})$ is a partial ultrametric space and $\cdot_{\sigma,\tau}$ is non-expansive.
\end{lemma}

\begin{proposition}
Let ${\cal D} = ( D,  \Omega^{{\cal D}}, \Lambda^{{\cal D}} ,  \dd{})$ be the structure where the  functions $\Lambda^{{\cal D}}_{\sigma, \tau}$ are  the interpretations of closed typed $\lambda$-terms on $D_{\infty}$.  Then  ${\cal D} $ is a partial quantitative extensional  $\lambda$-algebra.
\end{proposition}

\noindent {\bf Partial  $\lambda$-Algebras with Approximants.}  The two  examples above of partial quantitative $\lambda$-algebras can be viewed as special cases of a general construction, which can be carried out on any applicative algebra which includes projection operators. Namely, using the system of approximants given by the projection operators, one can endow the algebra with an ultrametric, getting a quantitative applicative algebra. If moreover the algebra satisfies  the $(\beta)$-rule, then it is a quantitative (weak) $\lambda$-algebra.

\begin{definition}[Applicative Algebra with Approximants] \label{aawa}
An applicative algebra $\C A = (A, \Omega^{\C A} ) $  has {\em approximants} if it includes 
 projection operators $\pi_{\sigma}^n: A_{\sigma} \to A_{\sigma}$, for all $\sigma \in T$ and $n\in {\mathbf N}$,  satisfying the following property:
\\ \hspace*{1.1cm} $
\mbox{for all } n\in {\mathbb N}, \mbox{\ for all } a,b\in A_{\sigma},  \ \pi^{n+1}_{\sigma}  (a)= \pi^{n+1}_{\sigma} (b) \mbox{ implies }  \pi^{n}_{\sigma}  (a)=  \pi^{n}_{\sigma} (b) \ .
$

\end{definition}

\begin{proposition}\label{propapp}
Let $\C A = (A, \Omega^{\C A}) $  be an applicative $\lambda$-algebra with approximants. Let $d^{\C A}_{\sigma}: A_{\sigma} \times A_{\sigma} \rightarrow {\mathbb R}_{\geq 0}$ be the family of functions defined  by induction on types
as follows: 
\\ \hspace*{2.8cm}  $d^{\C A}_{o}(a, b) =\begin{cases}  0 & \mbox{ if } a=b
 \\ 1 & \mbox{ otherwise }  \end{cases} $ \ \ \ and  \ \ \  $d^{\C A}_{\sigma\rightarrow \tau} (a,b) =
\frac{1}{2^m} \ , $ 
\\ where $m$ is the maximal $n\in {\mathbb N}$, if it exists, such that  
\\ 
\hspace*{0.2cm}  {\bf(1)}  $\pi_{\sigma\to \tau}^n(a) =\pi_{\sigma\to \tau}^n(b)$ \ \  and\ \  {\bf(2)}  $\forall c,d \in A_{\sigma}.\ (d^A_{\sigma} (c,d) \leq \frac{1}{2^n} \ \Longrightarrow \ d^A_{\tau} (ac,bd) \leq \frac{1}{2^n} )$ ;
\\
if the maximal  $n$ does not exist, then $d^A_{\sigma\rightarrow \tau} (a,b) =0$.

\noindent Then $ (A, \Omega^{\C A}, d^{\C A}) $  is a quantitative applicative  algebra.
\end{proposition}

 \longversion{Using an argument similar to that  for  Lemma~\ref{dnfpm}, one can  to prove that $d^A$ is a partial ultrametric.}


\section{Approximate Quantitative Algebras}\label{section7}
As we have seen, finding non-discrete quantitative (weak) $\lambda$-algebras is difficult. One difficulty arises from the non-expansiveness requirement on application. 
In Section~\ref{section5} we have shown how to define non-trivial  ultrametric quantitative $\lambda$-algebras, still maintaining the non-expansiveness requirement for application, but at the price of the partiality of the metric.
Here we present a different approach: we relax  rule (NExp) for application, so as to get  quantitative $\lambda$-algebras with  full pseudo-metric distances. 
Namely, we introduce the notion of {\em approximate applicative algebra}: this amounts to an applicative algebra with approximants (see Def.~\ref{aawa} above), and operators $\{ \cdot^n \}_{n\in {\mathbb N}}$ approximating application. Projection operators immediately induce an ultrametrics on the algebra, by just considering condition 1 in Prop.~\ref{propapp} (and dropping condition 2).  In general, application is expansive w.r.t. this metric (see  Remark~\ref{remn}). However, the milder {\em uniform
non-expansiveness} condition for approximant operators is satisfied in many cases, including the term algebra of normal forms and the $D_{\infty}$ model of Section~\ref{section5}. This approach is quite general, since it works both for the  typed and the untyped $\lambda$-calculus.

\begin{lemma}
Let $\C A =( A, \Omega^{\C A})$ be an applicative algebra with approximants. Let $e^{\C A}_{\sigma}: A_{\sigma} \times A_{\sigma} \rightarrow {\mathbb R}_{\geq 0}$ be the family of functions defined  by:
$ e^{\C A}_{\sigma} (a,b) = \frac{1}{2^m} \ , $
 where $m$ is the maximal $n\in {\mathbb N}$, if it exists, such that  $a_n=b_n$, otherwise we put $e^{\C A}_{\sigma} (a,b) =0$.
\\ Then for all $\sigma$ $(A, e^{\C A}_{\sigma} )$ is an ultrametric space.
\end{lemma}

\begin{definition}[Approximate  Quantitative Algebra]\hfill  \label{apprL}
\\ (i) An {\em approximate  algebra} $\C A =( A, \Omega^{\C A})$ is an applicative algebra with approximants whose signature includes also a family of operators  $\cdot^n_{\sigma,\tau}: A_{\sigma\to\tau} \times A_{\sigma} \to A_{\tau}$, for all $\sigma, \tau\in T$ and $n\in {\mathbf N}$ (the operators $\cdot^n_{\sigma,\tau}$ will be simply denoted by $\cdot^n$).
\\ (ii) An {\em approximate quantitative algebra} $\C A =( A, \Omega^{\C A}, e^{\C A})$ is an approximate algebra where the operators $\cdot_{\sigma,\tau}^n$ satisfy the following conditions:
\\
{\bf (1)}  for all $a\in A_{\sigma\to\tau}$, $b\in A_{\sigma}$, $n\in {\mathbf N}$, $e_{\tau}^A (a \cdot^{n+1}b, a \cdot b)\leq e_{\tau}^{\C A} (a \cdot^{n}b, a \cdot b)$;
\\ {\bf (2)}  for all $a\in A_{\sigma\to\tau}$, $b\in A_{\sigma}$, for all $\epsilon>0$ there exists $n\in {\mathbf N}$ such that $e_{\sigma}^{\C A} (a \cdot^n b, a\cdot b)\leq \epsilon$; 
\\ {\bf (3)}  (uniform non-expansiveness) $\forall n>0 \ \exists \epsilon_n >0$ s.t. $\forall \epsilon\leq \epsilon_n$,   $\forall a,b\in A_{\sigma\to\tau}$,  $\forall c,d \in A_{\sigma}$,
\\     \hspace*{2.3cm} $ e_{\sigma\to \tau }^{\C A} (a, b)\leq \epsilon \mbox{ and } e_{\sigma}^{\C A}(c, d)\leq \epsilon \mbox{ implies } e_{\sigma}^{\C A} (a\cdot^n c, b\cdot^n  d)\leq \epsilon$.
\end{definition}

Conditions 1 and 2 above  express the fact that the operators $\cdot^n$  approximate the behaviour of application; condition 3 replaces rule (NExp) for application.\medskip

\noindent {\bf The Approximate $\lambda$-Algebra of the Term Model.} The $\lambda$-algebra $\C {NF}$ can be extended to an approximate  quantitative $\lambda$-algebra by defining operators $\cdot^n_{\sigma,\tau}$ as follows:
\\  \hspace*{1.8cm} $\mbox{ for all } t\in NF_{\sigma\to\tau}, s\in NF_{\sigma},  \mbox{ for all } n\in {\mathbf N}, \ \ t\cdot^n_{\sigma, \tau}s =  t_n \cdot_{\sigma, \tau} s_n \ .  $
\\ One can check that the approximant operators satisfy all conditions of Definition~\ref{apprL}.
\medskip

\noindent {\bf The Approximate $\lambda$-Algebra of $D_{\infty}$.} 
The $\lambda$-algebra $\C D$ can be extended to an approximate quantitative $\lambda$-algebra by defining operators $\cdot^n_{\sigma,\tau}$ as follows:
\\  \hspace*{1.8cm} $  \mbox{ for all } a\in D_{\sigma\to\tau}, b\in D_{\sigma}, \mbox{ for all } n\in {\mathbf N}, \ \ a\cdot^n_{\sigma, \tau} b= a_{n+1}  \cdot_{\sigma, \tau} b \ .  $
\\ One can check that the approximant operators satisfy all conditions of Definition~\ref{apprL}. 
\\ Notice that the approximate algebra of $D_{\infty}$ yields a $\lambda$-algebra for the untyped $\lambda$-calculus.
\medskip

Finally, notice that in dealing with partial and approximate algebras we have considered applicative algebras over an extended signature. 
For lack of space, we have not developed corresponding approximate theories including extra operators and the suitable rules on them.
In particular, rule (NExp) has to be replaced by a rule expressing uniform non-expansiveness of approximant operators.
We leave this as future work; here we just observe that the appropriate
   language for reasoning on such structures would be the {\em indexed $\lambda$-calculus} together with {\em indexed reduction}, see \cite{Baren85}.

\section{Conclusions}\label{section8}
\subparagraph*{Contributions.}
This paper addresses the problem of defining quantitative algebras, in the sense of Mardare {\em et 
al}, capable to interpret terms of higher-order calculi. 
Our contributions include both negative and positive results: on the one hand we identify the main mathematical obstacles to the construction of non-trivial quantitative higher-order algebras; on the other hand we introduce a quantitative variant of the traditional notions of (weak) $\lambda$-algebras, together with a sound and complete equational syntax, and we show that,
in spite of the limitations highlighted, 
 intriguing notions of distance for the $\lambda$-calculus do indeed exist.
 
%
%

\subparagraph*{Related Work.}
Since \cite{ARNOLD1980181} metric spaces have been exploited as an alternative, quantitative,  framework to standard, domain-theoretic, denotational semantics \cite{VANBREUGEL20011,BAIER1994171}. 
The possibility of giving a metric structure to \emph{linear} or \emph{affine} higher-order 
programs is known, since $\Met$ is an SMCC, even if not a CCC. In this sense it is worth 
recalling the work by de Amorim {\em et al}~\cite{Gaboardi2017}, along with those of Reed and 
Pierce~\cite{Reed_2010}.
Moreover, 
ultra-metrics have already been used to model PCF \cite{Escardo1999}.
 More recently, Pistone has given a precise account of 
cartesian closed structure in categories of generalized metric spaces~\cite{Pistone2021}. In 
particular, it is known that if the quantale that captures distances can vary as the types 
vary, as for example in the so-called differential logical relations~\cite{dallago}, 
categories of generalized metric spaces can become cartesian closed. The study of metric 
semantics for imperative and concurrent programming languages has a long 
tradition~\cite{deBakker92,deBakker96}. However, this very sophisticated apparatus is not 
applicable to higher order programming languages.

Partial metrics have been well-studied since \cite{matthews} and \cite{Hohle:1992aa} (where they are called $M$-sets).  
\cite{Bukatin1997} shows that these metrics are strongly related to Scott semantics.
The setting of \emph{quantaloid-enriched} categories \cite{Stubbe2018, Stubbe2014} provides an abstract unifying framework for the different metric structures discussed here.
In this setting, \cite{Clementino2006, Clementino:2009aa} provide a general characterization
of exponentiable morphisms and objects in categories of (generalized) metric spaces.

Finally, a somehow related approach to quantitative reasoning is provided by the use of fuzzy logic to reason about degrees of similarity between programs, as spelled out in Zadeh's pioneering work \cite{Zadeh1, Zadeh2}.
More recently, fuzzy algebraic theories in the style of Mardare {\em et al}~have been studied \cite{Miculan2022}. However, such theories seem to lack a compositionality condition comparable to the one expressed by axiom (Nexp), hence apparently diverging from the idea of interpreting  
programs as non-expansive functions.

\subparagraph*{Perspectives.}
Our  
focus here was on the simplest case, namely that of algebras \emph{without} a barycentric structure, 
thus putting ourselves in a simpler setting than the one studied by Mardare {\em et al.}
Indeed, a natural continuation of this work would be to study quantitative algebras for $\lambda$-calculi enriched with operations having an intrinsically quantitative content, like e.g.~probabilistic choice (see \cite{dallago2017}) or some form of differentiation \cite{ER, Cai2014, Pearlmutter2016}. 
Another direction to explore, already suggested by some of our models, is that of exploring quantitative algebras in categories of domains like, e.g.~metric CPOs \cite{Gaboardi2017}, or continuous Scott domains (especially in virtue of their close connection with partial metric spaces \cite{Bukatin1997}).

Finally, the approaches of partial and approximate algebras open new lines of investigation: suitable approximate theories are called for, and moreover the distances on programs which arise are worth to be studied in depth.

\bibliography{main.bib}
%

\end{document}